\documentclass[a4paper]{IEEEtran}

\usepackage[utf8]{inputenc}
\usepackage[T1]{fontenc}
\usepackage{ifthen}
\usepackage[cmex10]{amsmath} %

\usepackage{tikz}
\usetikzlibrary{arrows,arrows.meta,matrix,positioning,positioning,calc,shapes}
\usepackage{pgfplots}
\pgfplotsset{compat=newest}
\usepgfplotslibrary{fillbetween}
\usepackage{mathtools}
\usepackage{amssymb}
\usepackage{amsthm}
\usepackage[capitalise]{cleveref}

\newtheorem{theorem}{Theorem}
\newtheorem{lemma}{Lemma}
\newtheorem{corollary}{Corollary}

\newtheorem{remark}{Remark}

\newtheorem{definition}{Definition}
\newtheorem{construction}{Construction}

\def\ve#1{{\mathchoice{\mbox{\boldmath$\displaystyle #1$}}%
              {\mbox{\boldmath$\textstyle #1$}}%
              {\mbox{\boldmath$\scriptstyle #1$}}%
              {\mbox{\boldmath$\scriptscriptstyle #1$}}}}

\usepackage{xcolor}

\newcommand{\F}{\ensuremath{\mathbb{F}}}
\newcommand{\code}{\ensuremath{\mathcal{C}}}
\newcommand{\Code}{\ensuremath{\mathcal{C}}}
\newcommand{\dmin}{\ensuremath{d_{\min}}}

\newcommand{\cL}{\ensuremath{\mathcal{L}}}
\newcommand{\cO}{\ensuremath{\mathcal{O}}}
\newcommand{\cR}{\ensuremath{\mathcal{R}}}
\newcommand{\cB}{\ensuremath{\mathcal{B}}}
\newcommand{\cA}{\ensuremath{\mathcal{A}}}
\newcommand{\cI}{\ensuremath{\mathcal{I}}}

\newcommand{\cW}{\ensuremath{\mathcal{W}}}
\newcommand{\bs}{\ensuremath{\mathbf{s}}}
\newcommand{\sbar}{\ensuremath{\bar{\bs}}}
\newcommand{\bB}{\mathbf{B}}
\newcommand{\bA}{\mathbf{A}}
\newcommand{\bF}{\mathbf{F}}
\newcommand{\bb}{\mathbf{b}}

\newcommand{\bx}{\mathbf{x}}

\newcommand{\bH}{\mathbf{H}}
\newcommand{\bC}{\mathbf{C}}
\newcommand{\0}{\ensuremath{\ve{0}}}
\newcommand{\bc}{\ensuremath{\ve{c}}}
\newcommand{\bu}{\ensuremath{\ve{u}}}
\newcommand{\alphaVec}{\ensuremath{\ve{\alpha}}}
\newcommand{\gammaVec}{\ensuremath{\ve{\gamma}}}

\newcommand{\family}{\ensuremath{\mathcal{F}}}
\newcommand{\pmds}{\ensuremath{\mathsf{PMDS}}}
\newcommand{\pmdsfamily}{\ensuremath{\family_\mathsf{PMDS}}}
\newcommand{\Gab}{\ensuremath{\mathsf{Gab}}}

\newcommand{\msr}{\ensuremath{\mathsf{MSR}}}
\newcommand{\sd}{\ensuremath{\mathsf{SD}}}
\newcommand{\bHa}{\ensuremath{\bH^{(a)}}}
\newcommand{\order}{\mathcal{O}}

\newcommand{\ZZ}{\mathbb{Z}}
\DeclareMathOperator{\rank}{rank}
\DeclareMathOperator{\diag}{diag}
\newcommand{\Fq}{\mathbb{F}_q}
\newcommand{\FqM}{\mathbb{F}_{q^M}}

\newcommand{\bG}{\mathbf{G}}
\newcommand{\bGa}{\bG^{(a)}}

\renewcommand{\bb}{\mathbf{b}}

\newcommand\qbin[3]{\left[\begin{matrix} #1 \\ #2 \end{matrix} \right]_{#3}}

\newcommand{\Clocal}{\code_{\mathsf{local}}}
\newcommand{\CGab}{\code_{\mathcal{G}}}
\newcommand{\CLRS}{\code_{\mathcal{LRS}}}

\newcommand{\svdots}{\raisebox{3pt}{$\scalebox{.75}{$\vdots$}$}}
\newcommand{\sddots}{\raisebox{3pt}{$\scalebox{.75}{$\ddots$}$}}

\newcommand{\QA}{Q_\mathsf{A}}
\newcommand{\QB}{Q_\mathsf{B}}
\newcommand{\QC}{Q_\mathsf{C}}
\newcommand{\QD}{Q_\mathsf{D}}
\newcommand{\QE}{Q_\mathsf{E}}

\usepackage{enumitem}
\usepackage{xspace}
\newcommand{\scrambledYB}{scrambled YB grouping }

\definecolor{constructionAcolor}{rgb}{1,0.4,0}  %
\definecolor{constructionBcolor}{rgb}{0.8,0,1}	%
\definecolor{constructionCcolor}{rgb}{0,0,1}	%
\definecolor{constructionDcolor}{rgb}{0,0,0}	%
\definecolor{constructionEcolor}{rgb}{0,0.7,0}	%

\title{Partial MDS Codes with Regeneration} %

\author{%
  \IEEEauthorblockN{Lukas Holzbaur, \emph{Student Member, IEEE},
                    Sven Puchinger, \emph{Member, IEEE},\\
                    Eitan Yaakobi, \emph{Senior Member, IEEE},
                    and Antonia Wachter-Zeh, \emph{Senior Member, IEEE}}\\
  \thanks{This work was partially presented at the IEEE International Symposium on Information Theory (ISIT) 2020 \cite{holzbaur2020partial}.

  The work of L.~Holzbaur and A.~Wachter-Zeh was supported by the German Research Foundation (Deutsche Forschungsgemeinschaft, DFG) under Grant No. WA3907/1-1. S.~Puchinger has received funding from the European Union’s Horizon 2020 research and innovation programme under the Marie Skłodowska-Curie grant agreement no.~713683. E.~Yaakobi was supported in part by the Israel Science Foundation under Grant No. 1817/18, by the Technion Hiroshi Fujiwara Cyber Security Research Center, and by the Israel National Cyber Directorate.
  This work was also supported by the Technical University of Munich -- Institute for Advanced Study, funded by the German Excellence Initiative and European Union 7th Framework Programme under Grant Agreement No. 291763.

  L.~Holzbaur, S.Puchinger, and A.~Wachter-Zeh are with the Institute for Communications Engineering, Technical University of Munich, Germany.
S.~Puchinger was with the Department of Applied Mathematics and Computer Science, Technical University of Denmark (DTU), Denmark.
  E.~Yaakobi is with the Computer Science Department, Technion --- Israel Institute of Technology, Israel.

  Emails: lukas.holzbaur@tum.de, sven.puchinger@tum.de, antonia.wachter-zeh@tum.de, yaakobi@cs.technion.ac.il}
}

\IEEEoverridecommandlockouts
\begin{document}

\maketitle

\begin{abstract}
  Partial MDS (PMDS) and sector-disk (SD) codes are classes of erasure correcting codes that combine locality with strong erasure correction capabilities. We construct PMDS and SD codes with local regeneration where each local code is a bandwidth-optimal regenerating MDS code. In the event of a node failure, these codes reduce both, the number of servers that have to be contacted as well as the amount of network traffic required for the repair process.
  The constructions require significantly smaller field size than the only other construction known in literature.
Further, we present a construction of PMDS codes with global regeneration which allow to efficiently repair patterns of node failures that exceed the local erasure correction capability of the code and thereby invoke repair across different local groups.
\end{abstract}

\section{Introduction}

Distributed data storage is ever increasing its importance with the amount of data stored by cloud service providers and data centers in general reaching staggering heights. The data is commonly spread over a number of nodes (servers or hard drives) in a \emph{distributed storage system} (DSS), with some additional redundancy to protect the system from data loss in the case of node failures (erasures). The resilience of a DSS against such events can be measured either by the minimal \emph{number of nodes} that needs to fail for data loss to occur, i.e., the \emph{distance} of the storage code, or by the expected time the system can be operated before a failure occurs that causes data loss, referred to as the \emph{mean time to data loss}. For both measures the use of maximum distance separable (MDS) codes provides the optimal trade-off between storage overhead and resilience to data loss (note that replication is a trivial MDS code). The downside of using MDS codes is the cost of recovering (replacing) a failed node. Consider a storage system with $k$ information nodes and $s$ nodes for redundancy. If an MDS code is used for the recovery of a node by means of erasure decoding, it  necessarily involves at least $k$ nodes (helpers) and, if done by straight-forward methods, a large amount of network traffic, namely the download of the entire content from $k$ nodes. To address these issues, the concepts of \emph{locally repairable codes} (LRCs) \cite{gopalan2012locality,kamath2014codes,rawat2014,Krishnan2018,gligoroski2017repair,hollmann2014minimum,li2016relieving} and \emph{regenerating codes} \cite{dimakis2010network,cadambe2013asymptotic,ye2017optimalRepair} have been introduced.

To lower the amount of network traffic in recovery, and thereby the required bandwidth, regenerating codes allow for repairing nodes by accessing $d > k$ nodes, but only retrieve a function of the data stored on each node. This significantly decreases the repair bandwidth, i.e., the amount of data that needs to be transmitted for the recovery of a number of failed nodes. Lower bounds on the required bandwidth for repair have been derived in \cite{dimakis2010network,cadambe2013asymptotic} which lead to two extremal code classes, namely \emph{minimum bandwidth regenerating} (MBR) and \emph{minimum storage regenerating} (MSR) \emph{codes}. MBR codes offer the lowest possible repair bandwidth, but at the cost of increased storage overhead compared to MDS codes. In this work we consider $d$-MSR codes, which require a higher bandwidth for repair than MBR codes, but are optimal in terms of storage overhead, i.e., they are MDS.

To address the other downside of node recovery in MDS codes, namely the large number of required helper nodes, LRCs introduce additional redundancy to the system, such that in the (more likely) case of a few node failures the recovery only involves less than $k$ helper nodes, i.e., can be performed \emph{locally}. This subset of helper nodes is referred to as a \emph{local code}. Recently several constructions of LRCs which maximize the distance have been proposed. However, when considering the mean time to data loss as the performance metric, distance-optimal LRCs are not necessarily optimal, as it is possible to tolerate many failure patterns involving a larger number of nodes than the number that can be guaranteed, while still fulfilling the locality constraints \cite{tamo2016optimal,holzbaur2019error}. \emph{Partial MDS} (PMDS) \emph{codes} \cite{blaum2013partial,blaum2016construction,gabrys2018constructions}, also referred to as \emph{maximally recoverable codes} \cite{chen2007maximally,gopalan2014explicit}, are a subclass of LRCs which guarantee to tolerate \emph{all} failure patterns possible under these constraints and thereby maximize the mean time to data loss. Specifically, an $(r,s)$-PMDS code of length $\mu n$ can be partitioned into $\mu$ local groups of size $n$, such that any erasure pattern with $r$ erasures in each local group plus any $s$ erasures in arbitrary positions can be recovered.

However, the local recovery of nodes still requires substantial network bandwidth, as the entire content of the helper nodes needs to be downloaded when considering straight-forward use recovery algorithms. To circumvent this bottleneck, several locally regenerating codes \cite{dimakis2010network} have been proposed \cite{kamath2014codes,rawat2014,Krishnan2018,gligoroski2017repair,hollmann2014minimum,li2016relieving}. In \cite{calis2016} it was shown that the LRC construction of \cite{rawat2014} is in fact a PMDS code, implicitly giving the first construction of PMDS codes with local regeneration\footnote{The construction in \cite{rawat2014} consists of two encoding stages, where in the second stage an arbitrary linear MDS code can be used to obtain the local codes. In \cite{calis2016} it was shown that the construction in fact gives a PMDS code, independent of the explicit choice of the MDS code in the second encoding stage. It follows that using a regenerating MDS code in the second encoding stage results in a PMDS code with local regeneration.}. However, these PMDS codes require a field size exponential in the length of the code and the subpacketization of the local regenerating code (which may itself be exponential in the length of the local code).
In the first part of this paper, we propose several constructions of locally MSR PMDS codes with significantly smaller field size than the construction in \cite{rawat2014}.

In the second part of this work, we consider PMDS codes with global regeneration properties that offer non-trivial repair schemes for the case where local recovery is not possible. Specifically, we give a PMDS code construction that, when punctured in any $r$ positions in each local group, becomes an MSR code. The reduction in global repair bandwidth is of particular interest, as the bandwidth of connections between nodes of different local groups is often assumed to be smaller than of nodes within the same local groups. Accordingly, though being less likely to occur, the non-local repair of a larger number of erasures can take a substantial amount of time. Repair problems where the communication cost within a local group differs from the cost of communication between local groups have been studied in \cite{gaston2013realistic,pernas2013non,sohn2018class,sohn2018capacity} and, in particular, the notion of \emph{rack-aware regenerating codes} (RRCs) \cite{hou2019rack,hou2020minimum} has been introduced. In this setting, the nodes are partitioned into a smaller number of racks, similar to the partitioning of nodes for codes with locality. Under this model, when a (number of) node(s) fails within a rack, it is regenerated by transmitting from each rack a function of the content of its nodes. The distinction to regenerating codes is that the repair bandwidth is given measured in terms of the amount of data transmitted \emph{between} the racks, while communication within each rack is ignored. Aside from this definition of the repair bandwidth, there are two important differences to the model we consider: 1) RRCs require a node that collects the data from the nodes within the rack and computes a function of it that is to be transmitted and 2) RRCs generally \emph{do not have locality}, i.e., no repair is possible within each rack.
\emph{Double regenerating codes} \cite{hu2016double} refine this model by considering two levels of regeneration, a local one, i.e., within the racks, and a global one, i.e., across the racks. %
A sightly different model has been considered in \cite{prakash2018storage}, in which repair is conducted by downloading a number of symbols from helper racks (also called clusters) and additionally a number of symbols is downloaded from a set of nodes of the same rack, where, unlike for RRCs, both contribute to the overall repair bandwidth. Similar to RCCs, the codes under this model do not have locality.

A rack-aware setting that also considers local recovery from node failures are codes for \emph{multi-rack distributed storage} \cite{tebbi2014code,qu2018multi}. There, a small number of nodes can be repaired/regenerated locally and failure patterns for which this is not possible are repaired by contacting other racks in addition to the surviving local nodes. Similar to RRCs, it is assumed that the contacted helper racks can process the data of the nodes within the rack and that the communication between racks is more costly than within a rack. Along with an information-theoretic bound, \cite{qu2018multi} presents a construction for the case of an efficient local repair of a single node failure. The minimization of the cross-rack repair bandwidth is stated as an open problem. In \cite{tebbi2014code} the authors consider a more general setting which improves both, the repair bandwidth within a rack in case of a small number of failure and across racks for failure patterns that cannot be repaired locally. Similar to RRCs, this model differs from the one in this work in that racks are able to process the data from their nodes prior to sending it to other racks. Additionally, we consider a stronger notion of locality in requiring the storage code to be PMDS.%

The work with closest relation to the model of global regeneration in codes with locality that we consider is \cite{gligoroski2017repair}, which introduces local redundancy by splitting parity-check equations of HashTag codes \cite{kralevska2016general,kralevska2017hashtag,li2018generic}. While it is shown that the codes are distance-optimal LRCs, they are generally not PMDS codes and possess only information locality, i.e., the recovery from a small subset of positions is only guaranteed for a set of systematic positions. Further, as the construction of HashTag codes \cite{kralevska2016general,kralevska2017hashtag} is not explicit, but partially relies on computer search, the construction of these parity-split HashTag codes with locality \cite{gligoroski2017repair} also partially relies on computer search.

\subsection{Contributions and Outline}

In \cref{sec:constructionS2}, we construct a new PMDS code with two global parities ($s=2$), where each local code is a $d$-MSR code. The construction is a non-trivial combination of the PMDS codes in \cite{blaum2016construction} with the MSR codes in \cite{ye2017optimalRepair}. This PMDS code construction is based on Reed-Solomon codes and only defined for a specific set of code locators. We generalize this construction to accept arbitrary code locators, which enables us to combine it with the MSR codes of \cite{ye2017optimalRepair}. This construction has field size in the order of
\begin{equation*}
O(\mu r^2 n).
\end{equation*}

In \cref{sec:universalPMDSconstruction}, we present a new general construction of locally MSR PMDS codes for any number of global parities. The construction is based on the observation that any universal PMDS code (that is, the local codes can be chosen almost arbitrarily) can be combined with a specific subclass of MSR codes, namely, MSR codes for which each row is an MDS code, to obtain locally MSR PMDS codes. The main contribution in this part of the work is the recognition of the interplay of these properties, both of which received little attention in the literature so far. This immediately leads to several new explicit locally MSR PMDS codes using known universal PMDS code families and the MSR codes in \cite{ye2017optimalRepair}: the PMDS codes in \cite{rawat2014} result in a field size in the order of
\begin{equation*}
O\big((rn)^{\mu (n-r)}\big)
\end{equation*}
and the ones in \cite{martinez2019universal} give a field size in
\begin{equation*}
O\big(\max\{rn,\mu+1\})^{n-r}\big).
\end{equation*}
We also slightly generalize the PMDS code family in \cite{gabrys2018constructions} and prove that this generalization in fact gives a universal PMDS code. The resulting field size of the corresponding locally MSR PMDS code is in
\begin{equation*}
O(nr(2n\mu)^{s(r+1)-1}).
\end{equation*}
All new locally MSR PMDS codes have the same subpacketization as the underlying MSR code %
from~\cite{ye2017optimalRepair}.

In \cref{sec:discussion}, we analyze the field size of the new constructions of locally MSR PMDS codes.
For the two-global-parities construction and the universal construction with the PMDS codes in \cite{martinez2019universal} and \cite{gabrys2018constructions}, there is a reasonable parameter range in which the respective construction has lowest field size among all known constructions.
Moreover, for all parameters, there is a new construction that has a smaller field size than the known construction in \cite{rawat2014}.

In \cref{sec:global_regeneration}, we propose the first known construction of globally MSR PMDS codes (that is, the MDS code obtained from puncturing $r$ positions in each local group is an MSR code), which allows for a significant reduction in the repair bandwidth in case a global repair event is triggered. To achieve this, we introduce a new MSR code construction based on~\cite{ye2017optimalRepair} which utilizes Gabidulin instead of Reed-Solomon codes and prove that it is in fact an MSR code. This allows for building PMDS codes with regenerating properties in a similar fashion as the Gabidulin-code-based PMDS code construction (without regeneration) in \cite{rawat2014}. The involved part for retaining the MSR property for any puncturing of $r$ positions in each local group is the choice of evaluation points of these Gabidulin codes. We present an explicit choice based on pairwise trivially intersecting subspaces and prove that it fulfills the required property for any such puncturing pattern.
The resulting code has a field size in $O(n^{\mu(n+s)})$ and subpacketization in $O((8n)^{\mu n (n+s)})$.

\section{Preliminaries}\label{sec:preliminaries}

\subsection{Notation}
We write $[a,b]$ for the set of integers $\{a,a+1,\ldots,b\}$ and $[b]$ if $a=1$.
For a set of integers $R \subseteq [n]$ and a code $\code$ of length $n$ we write $\code |_R$ for the code obtained by restricting $\code$ to the positions indexed by $R$, i.e., puncturing in the positions $[n]\setminus R$. For an element $\alpha \in \F$ we denote its order by $\cO(\alpha)$.
For an $a\times b$  matrix $\bB$ we denote by $\bB_{i,j}$ the entry in the $i$-th row and $j$-th column. For the $i$-th row/column we write $\bB_{i,:}$ and $\bB_{:,i}$, respectively. For a set $\mathcal{S}\subset [b]$, we denote by $\bB_{\mathcal{S}}$ the restriction of the matrix $\bB$ to the columns indexed by $\mathcal{S}$.
We denote the Gaussian binomial coefficient, i.e., the number of $k$-dimensional subspaces of $\Fq^n$, by
  \begin{equation*}
    \qbin{n}{k}{q}  = 
      \begin{cases}
        \frac{(1-q^n)(1-q^{n-1})\hdots(1-q^{n-k+1})}{(1-q)(1-q^2)\hdots(1-q^k)}, & k\leq n ,\\
        0, & k > n .
      \end{cases} 
  \end{equation*}

We denote a linear code of length $n$, dimension $k$, and distance $d_{\min}$ over a field $\Fq$ by $[n,k,d_{\min}]_q$. If the field size or minimum distance is not relevant, we sometimes omit the respective parameters and write $[n,k]$, $[n,k]_q$, or $[n,k,d_{\min}]$. Similarly, an $[n,k]$ RS code denotes a Reed--Solomon code of length $n$, dimension $k$ and minimum distance $n-k+1$ over a sufficiently large field.
For a code over $\F_{q^\ell}$ that is linear over $\F_q$ we write $[n,k,d_{\min};\ell]_q$, $[n,k,d_{\min};\ell]$, or $[n,k;\ell]$, respectively. The parameter $\ell$ is referred to as the subpacketization of the code and as each codeword of this code can be viewed as an array over $\F_q$ with $n$ columns and $\ell$ rows, we also refer to such codes as \emph{array codes}.

This work is largely based on the constructions of PMDS codes by Rawat \emph{et al.} \cite{rawat2014}, {Blaum
\emph{et al.}} \cite{blaum2016construction} {and Gabrys \emph{et al.}} \cite{gabrys2018constructions}, Mart{\'\i}nez-Pe{\~n}as--Kschischang \cite{martinez2019universal}, and the construction of MSR codes {by Ye and Barg} in \cite{ye2017optimalRepair}. Since the notations in these works are conflicting, i.e., the same symbols are used for different parameters of the codes,~\cref{tab:notation} provides an overview of the notation used in this work compared to these works.
\begin{table*}[htb]
  \centering
  \caption{An overview of the notation used in this work compared to the notation used in \cite{blaum2016construction,gabrys2018constructions,martinez2019universal,rawat2014,ye2017optimalRepair}. The largest benefit from this comparison is in \cref{sec:constructionS2,sec:universalPMDSconstruction,sec:discussion}, where we construct and discuss PMDS codes with local MSR codes. Therefore, the length and number of parities in the MSR code construction of \cite{ye2017optimalRepair} are matched with the parameters of the local codes in our work. Note that, in our notation, the length of the MSR code in \cref{sec:global_regeneration}, where we consider PMDS codes with global repair properties, is $\mu(n-r)$ and the number of parities is $s$.}
  \setlength{\tabcolsep}{10pt}
  \begin{tabular}{lcccccc}
    Description &\cite{blaum2016construction}&\cite{gabrys2018constructions} & \cite{martinez2019universal} & \cite{rawat2014} & \cite{ye2017optimalRepair} & This work \\ \hline
    Number of local groups & $r$ & $m$ & $g$ & $g$ & - & $\mu$ \\
    Length of local MSR code & $n$ & $n$ & $r+\delta-1$ & $r+\delta-1$ & $n$ & $n$ \\
    Number of local parity symbols & $m$ & $r$ & $\delta-1$ & $\delta-1$ & $r$ & $r$ \\
    Number of global parity symbols & $s$ & $s$ & $h$ & $D-1$ & - & $s$ \\
    Code length & $rn$ & $mn$ & $n$ & $n$ & - & $\mu n$ \\
    Subpacketization & - & - & - & $\alpha$ & $l$& $\ell$ \\
    Number of nodes needed for repair & - & - & - & $d$ & $d$ & $d$
  \end{tabular}
  \label{tab:notation}
\end{table*}

\subsection{Definitions}

All code construction presented in the following are vector codes, where each row is an arbitrary codeword of a specific code. To keep the presentation compact, we define a notation for the special case of the code being the same in each row.

\begin{definition}\label{def:cartesianCodeProduct}
For a linear $[n,k,\dmin]$ code $\code$ over $\Fq$ we denote by
\begin{equation*}
  \code^{\times \ell} = \underbrace{\code \times \cdots \times \code}_{\ell \ \text{times}}
\end{equation*}
the $\ell$-fold Cartesian product of the code $\code$ with itself arranged as an $\ell \times n$ matrix, i.e., the set
\begin{equation*}
  \code^{\times \ell} = \{\bC \in \F_q^{\ell\times n} \ | \ \bC_{i,:} \in \code \ \forall \ i \in [\ell]\} \ .
\end{equation*}
\end{definition}

It is generally desirable to keep the size of the field in which operations are conducted small. While we have to rely on larger fields to achieve some of the code properties in the following, it can be useful to regard codes over larger fields as array codes over a subfield.
\begin{corollary}\label{col:cartesianCodeParameters}
Let $\{\gamma_1,\ldots,\gamma_\ell\}$ be a basis of $\F_{q^\ell}$ over $\Fq$ and $\code$ be an $[n,k,\dmin]$ code over $\Fq$. Then the code
\begin{align*}
  \left\{ (\gamma_1,\ldots, \gamma_\ell) \cdot \bC \ | \ \bC\in \code^{\times \ell} \right\} \simeq \langle \code \rangle_{\F_{q^\ell}}
\end{align*}
 is an $[n,k,\dmin]$ code over $\F_{q^\ell}$.
\end{corollary}
With these basic notions established, we now define the code classes and concepts used in this work, starting with a formal definition of PMDS codes.
This special class of LRCs is able to correct all patterns of erasures that are theoretically correctable and thereby provides the strongest possible protection against data loss, given the locality constraints.

\begin{definition}[Partial MDS array codes]\label{def:pmds}
Let $n,\mu,r,s,\ell \in \mathbb{Z}_{>0}$ be such that $\mu \geq 2$, $r < n$, and $s \leq (n-r)(\mu-1)$.
Let $\cW = \{W_1,W_2,\ldots,W_{\mu}\}$ be a partition of $[\mu n]$ with $|W_i|=n \ \forall \ i\in [\mu]$.

Let $\code \subset \F_q^{\ell \times \mu n}$ be a linear $[\mu n,(n-r)\mu -s;\ell]$ code. The code $\code$ is a $\pmds(\mu,n,r,s,\cW;\ell)$ \emph{partial MDS array code} if
\begin{itemize}
\item the code $\code |_{W_i}$ is an $[n,n-r,r+1; \ell]$ MDS code
for all $i \in [\mu]$ and
\item for any $E_i \subset W_i$ with $|E_i|=r \ \forall \ i\in [\mu]$, the code $\code |_{[\mu n] \setminus \cup_{i=1}^{\mu} E_i}$ is an $[\mu n-r\mu,\mu n-r\mu-s,s+1;\ell]$ MDS code.
\end{itemize}
\end{definition}

We refer to the code $\code |_{W_i}$ as the $i$-th \emph{local code}.
Furthermore, we refer to parameters $n,\mu,r,s$ satisfying the constraints of \cref{def:pmds} as \emph{valid PMDS parameters}.
Note that we exclude trivial parameters for which the definition results simply in an MDS code or a concatenation of MDS codes\footnote{There are several parameter combinations for which such a trivial case occurs.
  One trivial case is given by $r=0$ (and arbitrary $\mu,s$) where the code obtained from ``puncturing $r=0$ positions'' in each local group, i.e., the unpunctured code, is MDS by the second property of \cref{def:pmds}. On the other hand, if $s=0$ (and arbitrary $\mu,r$) the code is just a concatenation of (independent) local MDS codes.}. The requirement $s \leq (n-r)(\mu-1)$ is necessary for the PMDS code definition since otherwise the dimension of the local code exceeds the one of the global code --- a contradiction.

\begin{remark}
In \cite{blaum2016construction,gabrys2018constructions} each codeword of the PMDS and SD codes is regarded as a $\mu \times n$ array, where for PMDS codes each row can correct $r$ erasures and for SD codes $r$ erased columns can be corrected. As we will construct PMDS and SD codes with local MSR codes, we will require subpacketization, i.e., each node will not store a symbol, but a vector of multiple symbols. To avoid having different types of rows, we adopt the terminology commonly used in the LRC literature and view the codewords of a PMDS or SD code as vectors, and what we refer to as \textbf{local codes} is equivalent to the \emph{rows} of \cite{blaum2016construction,gabrys2018constructions}.
\end{remark}

In the following we will construct both PMDS and SD codes with local regeneration, but since the concepts and proofs are mostly the same, we provide them in less detail for SD codes.
\begin{remark}\label{rem:SDcodes}
  A Sector-Disk $\sd(\mu,n,r,s,\cW;\ell)$ code is defined similar to a PMDS codes as in~\cref{def:pmds}, except that $E_1=E_2= \cdots =E_{\mu}$ holds.
\end{remark}

\begin{definition}[Regenerating code \cite{dimakis2010network,cadambe2013asymptotic}]\label{def:regeneratingCode}
  Let $\mathcal{F}, \cR \subset [n]$ be two disjoint subsets.
  Let $\code$ be an $[n,n-r;\ell]$ MDS array code $\code$ over $\F_q$. Define $M(\code,\mathcal{F},\cR)$ as the smallest number of symbols of $\F_q$ one needs to download from the surviving nodes indexed by $\cR$ to recover the erased nodes indexed by $\mathcal{F}$. Then
  \begin{equation}\label{eq:boundRegenrating}
    M(\code,\mathcal{F},\cR) \geq \frac{|\mathcal{F}||\cR|\ell}{|\mathcal{F}|+|\cR|-n+r} \ .
  \end{equation}
  For two integers $h,d$, with $1 \leq h \leq r$ and $n-r \leq d \leq n-h$, we say that the code $\code$ is an \emph{$(h,d)$-MSR} code if
  \begin{equation*}
    \max_{\substack{|\mathcal{F}| = h, |\cR|=d \\ \mathcal{F} \cap \cR = \emptyset}} M(\code,\mathcal{F},\cR) = \frac{hd\ell}{h+d-n+r} \ .
  \end{equation*}
  If $h=1$ we say that the code is a $d$-MSR code. If in addition, $d=n-1$, we simply say that the code is an MSR code.
\end{definition}

Informally, in a regenerating array code, as in~\cref{def:regeneratingCode}, we require that every codeword can be recovered from an arbitrary subset of $n-r$ columns. We now define a slightly stronger property, which contains additionally a similar requirement on every row of a codeword.
\begin{definition}\label{def:rowWiseMDS}
Let $\code$ be an $[n,n-r;\ell]$ regenerating code as in~\cref{def:regeneratingCode}. We say that the code $\code$ is a \textbf{row-wise MDS} regenerating code if for any $i\in [\ell]$ the set $\left\{ \ve{C}_{i,:} \ | \ \ve{C} \in \code \right\}$ is an MDS code.
\end{definition}

With these notions established, we combine \cref{def:pmds,def:regeneratingCode} to formally define the class of codes we construct and analyse in \cref{sec:constructionS2,sec:universalPMDSconstruction,sec:discussion}.
\begin{definition}[Locally $(h,d)$-MSR PMDS array codes]\label{def:locallyMSR}
  Let $\code$ be a $\pmds(\mu,n,r,s,\cW;\ell)$ code and $d,h$ be chosen such that $1 \leq h$ and $n-r \leq d \leq n-h$. We say that the code $\code$ is locally $(h,d)$-MSR if $\code |_{W_i}$ is an $(h,d)$-MSR code for all $i\in [\mu]$.
  If $h=1$ we say the code is a locally $d$-MSR PMDS code. If in addition, $d=n-1$, we simply say that the code is an MSR PMDS code.
\end{definition}

\begin{figure}
  \centering
    \resizebox{\columnwidth}{!}{\def\x{0.5}%

\begin{tikzpicture}

\node (S1) at (0,0) [draw,thick,minimum width=\x*0.75cm,minimum height=\x*6.5cm] {};
\node (S2)  [right=\x*0.3cm of S1, draw,thick,minimum width=\x*0.75cm,minimum height=\x*6.5cm] {};
\node (S3)  [right=\x*0.3cm of S2, draw,thick,minimum width=\x*0.75cm,minimum height=\x*6.5cm] {};
\node (S32)  [right=\x*0.3cm of S3, draw,thick,minimum width=\x*0.75cm,minimum height=\x*6.5cm] {};
\node (S4)  [right=\x*0.3cm of S32, draw,thick,minimum width=\x*0.75cm,minimum height=\x*6.5cm] {};

\node (S5)  [right=\x*0.7cm of S4, draw,thick,minimum width=\x*0.75cm,minimum height=\x*6.5cm] {};
\node (S6)  [right=\x*0.3cm of S5, draw,thick,minimum width=\x*0.75cm,minimum height=\x*6.5cm] {};
\node (S7)  [right=\x*0.3cm of S6, draw,thick,minimum width=\x*0.75cm,minimum height=\x*6.5cm] {};
\node (S72)  [right=\x*0.3cm of S7, draw,thick,minimum width=\x*0.75cm,minimum height=\x*6.5cm] {};
\node (S8)  [right=\x*0.3cm of S72, draw,thick,minimum width=\x*0.75cm,minimum height=\x*6.5cm] {};

\node (S9)  [right=\x*0.7cm of S8, draw,thick,minimum width=\x*0.75cm,minimum height=\x*6.5cm] {};
\node (S10)  [right=\x*0.3cm of S9, draw,thick,minimum width=\x*0.75cm,minimum height=\x*6.5cm] {};
\node (S11)  [right=\x*0.3cm of S10, draw,thick,minimum width=\x*0.75cm,minimum height=\x*6.5cm] {};
\node (S112)  [right=\x*0.3cm of S11, draw,thick,minimum width=\x*0.75cm,minimum height=\x*6.5cm] {};
\node (S12)  [right=\x*0.3cm of S112, draw,thick,minimum width=\x*0.75cm,minimum height=\x*6.5cm] {};

\draw[dotted,thick] (\x*5.2, \x*3) -- (\x*5.2,-\x*3);
\draw[dotted,thick] (\x*11.1, \x*3) -- (\x*11.1,-\x*3);

\foreach \i in {1,...,12,32,72,112}{
  \foreach \j in {0,1,2,3}{
  \draw ($(S\i)+(-\x*0.25,\x*3.15-\x*\j*0.4)$) rectangle ($(S\i)+(+\x*0.25,\x*2.85-\x*\j*0.4)$) node (C\i) {};}
  \foreach \j in {0,1,2,3}{
  \draw ($(S\i)+(-\x*0.25,\x*1.35-\x*\j*0.4)$) rectangle ($(S\i)+(+\x*0.25,\x*1.05-\x*\j*0.4)$) node (C\i) {};}
  \foreach \j in {0,1,2,3}{
  \draw ($(S\i)+(-\x*0.25,-1.65*\x-\x*\j*0.4)$) rectangle ($(S\i)+(+\x*0.25,-1.95*\x-\x*\j*0.4)$) node (C\i) {};}
  \node (C\i) at ($(S\i)+(0,\x*-0.7)$) [minimum width=\x*0.5cm,minimum height=\x*0.2cm,rounded corners=1pt] {$\vdots$};

}

\draw[dashed, rounded corners = 1pt, blue, thick] ($(S1)+(-\x*.5,\x*3.3)$) rectangle ($(S12)+(\x*.5,\x*2.7)$);

\node[anchor = south east] (L1) at ($(S12)+(\x*1,\x*4.5)$) {\footnotesize PMDS codeword};
\path (L1.south east) edge[bend left, -{Latex[length=1mm,width=0.8mm]}]  ($(S12)+(\x*.7,\x*3.3)$) ;

\draw[dashed, rounded corners = 1pt, orange, thick] ($(S1)+(-\x*.5,\x*3.3)$) rectangle ($(S4)+(\x*.5,\x*1.5)$);
\node[anchor = south east] (L2) at ($(S4)+(\x*0.8,\x*4.4)$) {\footnotesize Row-wise MDS MSR codeword};
\path (L2.south east) edge[bend left, -{Latex[length=1mm,width=0.8mm]}]  ($(S4)+(\x*0.55,\x*2.3)$) ;

\draw [decorate,decoration={brace,amplitude=2pt}]  ($(S1)+(-\x*0.7,\x*1.5)$) -- ($(S1)+(-\x*0.7,\x*3.3)$) node [black,midway,xshift=-0.25cm] {\footnotesize $\ell$ };

\draw[thick,-{Latex[length=2mm,width=1mm]},red] ($(S2)+(0,\x*3)$) -- ($(S2)+(-\x*0.2,-\x*0.4)$) -- ($(S2)+ (\x*0.2,\x*0.4)$) -- ($(S2)+(0,-\x*3)$);
\draw[thick,-{Latex[length=2mm,width=1mm]},red] ($(S7)+(0,\x*3)$) -- ($(S7)+(-\x*0.2,-\x*0.4)$) -- ($(S7)+ (\x*0.2,\x*0.4)$) -- ($(S7)+(0,-\x*3)$);
\draw[thick,-{Latex[length=2mm,width=1mm]},red] ($(S9)+(0,\x*3)$) -- ($(S9)+(-\x*0.2,-\x*0.4)$) -- ($(S9)+ (\x*0.2,\x*0.4)$) -- ($(S9)+(0,-\x*3)$);
\draw[thick,-{Latex[length=2mm,width=1mm]},red] ($(S10)+(0,\x*3)$) -- ($(S10)+(-\x*0.2,-\x*0.4)$) -- ($(S10)+ (\x*0.2,\x*0.4)$) -- ($(S10)+(0,-\x*3)$);
\draw[thick,-{Latex[length=2mm,width=1mm]},red] ($(S11)+(0,\x*3)$) -- ($(S11)+(-\x*0.2,-\x*0.4)$) -- ($(S11)+ (\x*0.2,\x*0.4)$) -- ($(S11)+(0,-\x*3)$);
\draw[thick,-{Latex[length=2mm,width=1mm]},red] ($(S4)+(0,\x*3)$) -- ($(S4)+(-\x*0.2,-\x*0.4)$) -- ($(S4)+ (\x*0.2,\x*0.4)$) -- ($(S4)+(0,-\x*3)$);
\draw[thick,-{Latex[length=2mm,width=1mm]},red] ($(S112)+(0,\x*3)$) -- ($(S112)+(-\x*0.2,-\x*0.4)$) -- ($(S112)+ (\x*0.2,\x*0.4)$) -- ($(S112)+(0,-\x*3)$);

\node at ($(S6)+(\x*1.1,\x*4.7)$) {Servers};
\node at ($(S2)+(\x*1.1,\x*3.8)$) {\footnotesize Local group $1$};
\node at ($(S6)+(\x*1.1,\x*3.8)$) {\footnotesize Local group $2$};
\node at ($(S10)+(\x*1.1,\x*3.8)$) {\footnotesize Local group $3$};

\node (El) at ($(S6)+(\x*0.6,-\x*4.3)$) {\color{red} Server Failures};
\path ($(El.north west)+(\x*0,\x*-0.6)$) edge[bend left=5, -{Latex[length=1mm,width=0.8mm]}]  ($(S2)+(\x*0.1,-\x*3.4)$) ;
\path ($(El.north west)+(\x*0,\x*-0.2)$) edge[bend left=5, -{Latex[length=1mm,width=0.8mm]}]  ($(S4)+(\x*0.1,-\x*3.4)$) ;
\path ($(El.north) + (\x*0.5,0)$) edge[bend right=5, -{Latex[length=1mm,width=0.8mm]}]  ($(S7)+(-\x*0,-\x*3.4)$) ;
\path ($(El.north east)+ (\x*0,-\x*0.3)$) edge[bend right=5, -{Latex[length=1mm,width=0.8mm]}]  ($(S9)+(\x*0.1,-\x*3.4)$) ;
\path ($(El.north east)+ (\x*0,-\x*0.4)$) edge[bend right=5, -{Latex[length=1mm,width=0.8mm]}]  ($(S10)+(\x*0.1,-\x*3.4)$) ;
\path ($(El.north east)+ (\x*0,-\x*0.5)$) edge[bend right=5, -{Latex[length=1mm,width=0.8mm]}]  ($(S11)+(\x*0.1,-\x*3.4)$) ;
\path ($(El.north east)+(\x*0,\x*-0.6)$) edge[bend right=5, -{Latex[length=1mm,width=0.8mm]}]  ($(S112)+(-\x*0.1,-\x*3.4)$) ;

\end{tikzpicture}}
  \caption{Illustration of locally MSR PMDS array codes as constructed in this work, with $n=5$, $\mu=3$ and each symbol of the code alphabet represented by a small rectangle. The shown erasure pattern can be corrected by an $(r=2,s=2)$-PMDS code.}
  \label{fig:illustration}
\end{figure}
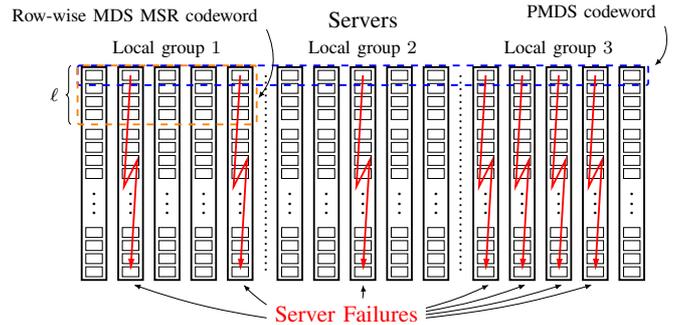
\cref{fig:illustration} shows an illustration of a locally MSR PMDS array code. Assuming it to be an $(r=2,s=2)$-PMDS code, the erasures in the first local code can be corrected locally, but without taking advantage of the regenerating property, as the number of available helper nodes is only $n-r$. The erasure in the second local code can be corrected from the remaining $n-r+1$ nodes in the local group using the locally regenerating property, and the erasures in the third local code can be recovered by accessing nodes of the other local groups. Note that the example was chosen specifically to illustrate these different cases, while the case of a single erasure in a local code, for which the locally regenerating property decreases the repair bandwidth, is far more likely than the other cases.

A globally-MSR PMDS code is formally defined as follows.
\begin{definition}[Globally $(h,d)$-MSR PMDS Code]\label{def:globallyMSR_PMDS}
Let $\code$ be a $\pmds(\mu,n,r,s,\cW;\ell)$ code and $d,h$ be chosen such that $1 \leq h \leq s$ and $\mu(n-r)-s \leq d \leq \mu(n-r)-h$.
We say that the code $\code$ is
\emph{globally $(h,d)$-MSR} if the restriction $\code |_{[\mu n] \setminus \cup_{i=1}^{\mu} E_i}$ is a $[\mu (n-r),\mu n-r\mu-s,s+1;\ell]$ $(h,d)$-MSR code
for any $E_i \subset W_i$ with $|E_i|=r$ for all $i\in [\mu]$. Again, we say $d$-MSR if $h=1$ and simply MSR if in addition $d=\mu(n-r)-1$.
\end{definition}

Throughout the paper, we consider in all constructions $h=1$.
This is the most interesting case since in a storage system, it is more likely that one node needs to be regenerated than multiple nodes.
In the globally-MSR case, we further fix $d$ to be maximal, i.e., $d=\mu(n-r)-1$. It can be seen from the bound in Definition~\ref{def:regeneratingCode} that the repair bandwidth decreases in $d$, i.e., it is minimal for this choice of $d$.
See Section~\ref{sec:conclusion} for a discussion on how the results can be generalized.

\subsection{Ye-Barg Regenerating Codes}

We repeat \cite[Construction~2]{ye2017optimalRepair} in the slightly different notation which will be used in this work.

\begin{definition}[Ye-Barg $d$-MSR codes {\cite[Construction~2]{ye2017optimalRepair}}] \label{def:yeBarg}
  Let $\code \subset \F_q$ be an $[n,n-r;\ell]$ array code over $\F_q$, where $q\geq b n$ and $b=d+1-n+r$. Let $\{\beta_{i,j} \}_{i\in [b], j\in [n]}$ be a set of $bn$ distinct elements of $\F_q$. Then each codeword is an array with $\ell = b^{n}$ rows and $n$ columns, where the $a$-th row fulfills the parity check equations
  \begin{equation*}
    \bHa =
      \begin{bmatrix}
        1&1& \hdots & 1 \\
        \beta_{a_1,1} & \beta_{a_2,2} & \hdots & \beta_{a_{n},n}\\
        \vphantom{\int\limits^x}\smash{\vdots} & \vphantom{\int\limits^x}\smash{\vdots} & & \vphantom{1}\smash{\vdots} \\
        \beta_{a_1,1}^{r-1} & \beta_{a_2,2}^{r-1} & \hdots & \beta_{a_{n},n}^{r-1}
      \end{bmatrix}\ ,
  \end{equation*}
  for $a \in [0,\ell-1]$ and $a = \sum_{i = 1}^{n} a_i b^{i-1}$ with $a_i \in [0,b-1]$.

\end{definition}
It is easy to see that a Ye-Barg code as in~\cref{def:yeBarg} is in fact \emph{row-wise MDS} $d$-MSR code.

\begin{remark}
  The constructions presented in~\cref{sec:constructionS2,sec:universalPMDSconstruction} can also be applied to obtain locally $(h,d)$-MSR PMDS codes, where each local code is an $(h,d)$-MSR code as in \cite[Construction~3]{ye2017optimalRepair}, which is very similar in structure to \cite[Construction~2]{ye2017optimalRepair} given in~\cref{def:yeBarg}. However, as the required subpacketization is larger for the former, we focus on $d$-MSR codes in this work. %
\end{remark}
\begin{remark}
  In~\cref{def:yeBarg} we define each row of the array code by a set of parity check equations independent of the other $\ell-1$ rows of the array. Note that this is not possible for array codes in general. However, for the existence of such a description it is sufficient that the matrices $A_i$, as defined in \cite{ye2017optimalRepair}, are diagonal matrices. This simplifies the notation for the cases considered in this work, as this notation makes it obvious that each row is an $[n,n-r]$ RS code, and thus MDS.
\end{remark}

\subsection{Gabidulin Codes}\label{ssec:Gabidulin_codes}

Gabidulin codes \cite{Delsarte_1978,Gabidulin_TheoryOfCodes_1985,Roth_RankCodes_1991} are rank-metric codes that have been used repeatedly in the literature to construct LRCs and PMDS codes. They are defined as follows.

\begin{definition}[Gabidulin codes]\label{def:GabidulinCodes}
Let $\bb = (\beta_1,\ldots, \beta_n) \in \FqM^{n}$ be such that the $\beta_i$ are linearly independent over $\Fq$.
The $[n,k]_{q^M}$ Gabidulin code $\Gab(n,k,\bb)$ is defined as
\begin{align*}
\Gab(n,k,\bb) = \left\{ \bc \in \FqM^{n} \ | \ \bc \cdot \bH^T = \boldsymbol{0}\right\}
\end{align*}
with
\begin{align*}
\bH =
\begin{bmatrix}
    \beta_{1} & \beta_{2} & \hdots & \beta_{n}\\
    \beta_{1}^{q^1} & \beta_{2}^{q^1} & \hdots & \beta_{n}^{q^1}\\
    \vphantom{\int\limits^x}\smash{\vdots} & \vphantom{\int\limits^x}\smash{\vdots} & & \vphantom{\int\limits^x}\smash{\vdots} \\
    \beta_{1}^{q^{n-k-1}} & \beta_{2}^{q^{n-k-1}} & \hdots & \beta_{n}^{q^{n-k-1}}
\end{bmatrix} \ .
\end{align*}
\end{definition}
Note that the existence of linearly independent $\beta_i$ implies $n \leq M$.
In the paper, we refer to the set $\{\beta_i\}, i\in [n]$ as the \emph{code locators} of the Gabidulin code. Note that, unlike the code locators of GRS codes, the code locators of a Gabidulin code are generally not the same in the generator and the parity-check matrix. In the following, when we refer to the code locators of a Gabidulin code, we always refer to the $\beta_i$'s used for the \emph{parity-check matrix} as in~\cref{def:GabidulinCodes}.

The codewords of an $[n,k]_{q^M}$ Gabidulin code can be seen as matrices in $\Fq^{M \times n}$ by expanding elements of $\FqM$ into vectors in $\Fq^M$ (using a fixed basis of $\FqM$ over $\Fq$).
Thus, we can define the rank distance of two codewords as the rank of their matrix representations' difference.
It is well-known that the minimum rank distance of a Gabidulin code is $n-k+1$, i.e., it fulfills the Singleton-like bound in the rank metric with equality.

\section{Regenerating PMDS and Sector-Disk codes with Two Global Parities} \label{sec:constructionS2}

We construct array codes from the PMDS codes of \cite{blaum2016construction} using the ideas of \cite{ye2017optimalRepair} to obtain locally $d$-MSR PMDS codes. Since the PMDS code construction in \cite{blaum2016construction} can be easily turned into an SD code (which is a slightly weaker notion, but results in a smaller field size), we also include the respective construction of SD codes with local $d$-MSR codes in this section.

\subsection{Generalization of known PMDS construction} \label{sec:BlaumSDandPMDS}

 To apply the ideas of \cite{ye2017optimalRepair} when constructing locally $d$-MSR PMDS and SD codes we need the local codes to be RS codes with specific code locators. The construction of PMDS codes given in \cite{blaum2016construction} has the property that the local codes are RS codes, but the code locators are fixed to be the first $n$ powers of some element $\beta$ of sufficient order. We generalize this construction to allow for different choices of code locators for the local codes.

  Let $\beta \in \F_{2^w}$ be an element of order $\order(\beta) \geq \mu N$. The $[\mu n,\mu(n-r)-2]$ code $\code(\mu ,n,r,2,\mathcal{L},N)$ is given by the $(r\mu +2)\times \mu n$ parity-check matrix %
  \begin{align}
   \bH =
   \begin{bmatrix}
      \bH_0 & \mathbf{0} &\hdots & \mathbf{0} \\
      \mathbf{0} & \bH_0 &\hdots & \mathbf{0} \\
      \vdots & \vdots &\ddots & \vdots \\
      \mathbf{0}& \mathbf{0} &\hdots & \bH_0 \\
      \bH_1 & \bH_2 &\hdots & \bH_{\mu } \\
    \end{bmatrix}, \label{eq:PMDS_two_parities_parity_check_matrix}
  \end{align}
  where
  \begin{align*}
    \bH_0  =
 \begin{bmatrix}
     1 & 1 & \hdots & 1 \\
     \beta^{i_1} & \beta^{i_2} & \hdots & \beta^{i_{n}} \\
     \beta^{2i_1} & \beta^{2i_2} & \hdots & \beta^{2i_{n}}\\
     \vdots & \vdots & \ddots & \vdots \\
     \beta^{(r-1)i_1} & \beta^{(r-1)i_2} & \hdots & \beta^{(r-1)i_{n}}
    \end{bmatrix}
  \end{align*}
  for $\cL = \{i_1,i_2,\ldots,i_{n}\}$ and, for $0\leq j\leq \mu-1$,
  \begin{align*}
    \bH_{j+1} =
     \begin{bmatrix}
      \beta^{ri_1} & \beta^{ri_2} &  \hdots &\beta^{ri_n} \\
      \beta^{-jN-i_1} & \beta^{-jN-i_2} & \hdots & \beta^{-jN-i_{n}}
    \end{bmatrix} \ .
  \end{align*}

Note that this generalization includes both \cite[Construction~A]{blaum2016construction} and \cite[Construction~B]{blaum2016construction} as special cases:
\begin{equation*}
\code_A=\code(\mu ,n,r,2,\{0,1,\ldots, n-1\},n)
\end{equation*}
and
\begin{equation*}
\code_B=\code(\mu ,n,r,2,\{0,1,\ldots, n-1\},N_B)
\end{equation*}
for $N_B = (r+1)(n-1-r)+1$.

We now derive a general, sufficient condition on $N$, based on the set $\cL$, such that the code is a PMDS code.

\begin{lemma} \label{lem:PMDSgeneral}
  Let $\mu,n,r$ and $s=2$ be valid PMDS parameters and $\cL$ be a set of non-negative integers with $|\cL| = n$. Then, the code $\code(\mu ,n,r,2,\cL,N)$ is a PMDS code for any $N \geq (r+1) (\max_{i \in \cL}i-r)+1$.
\end{lemma}
\begin{IEEEproof}
  We follow the proofs of \cite[Theorem~5]{blaum2016construction} and \cite[Theorem~7]{blaum2016construction}.
  The difference to the construction above is that in \cite{blaum2016construction}, the powers $i_1,\dots,i_n$ are consecutive, i.e., $i_j = j-1$. This results in a slightly more technical proof.

  Assume $r$ positions in each local group (row of the PMDS code) have been erased and in addition there are $2$ random erasures. If the two erasures occur in the same local group $z$, all local groups except for this one will be corrected by the local codes. Assume the erasures in local group $z$ occurred in positions $\mathcal{E}_z \subset [n]$. Since all points in $\cL$ are distinct, by the same argument as in \cite{blaum2016construction}, the erased positions can be recovered uniquely if the matrix
  \begin{equation*}
   \hat{\bH} =
     \begin{pmatrix}
       \bH_0 \\ \bH_z
     \end{pmatrix}
  \end{equation*}
  restricted to the erased positions $\mathcal{E}$ is of full rank.
  Say that the erased positions are $1 \leq j_1 < j_2 < \cdots < j_{r+2} \leq n$.
  It is easy to see that this matrix $\hat{\bH}_{\mathcal{E}}$ can be transformed into a Vandermonde matrix by multiplying the last row by $\beta^{jN}$ and the $\xi$-th column by $\beta^{i_{j_\xi}}$ for all $\xi\in [r-2]$ (instead of $\beta^{j_\xi-1}$ as in \cite[Theorem~5]{blaum2016construction}).
  Therefore, it is of full rank and the erasures can be corrected. \\

  Now consider the case of two local groups (horizontal codes) with $r+1$ erasures each. Assume, without loss of generality, that the erased positions are given by $\{j_1,\ldots,j_{r+1}\} \subset \cL$ in local group $1$ and $\{j_1',\ldots,j_{r+1}'\} \subset \cL$ in local group $z+1$ with $1\leq z \leq \mu-1$. Define the matrix
  \begin{align*}
    &\bF(j_1,\ldots,j_{r+1};j_1,\ldots,j_{r+1}; r ; N ; z) = \\
    & \quad \begin{bsmallmatrix}
      1 &\hdots & 1 & 0 &\hdots & 0\\
      \alpha^{j_1} & \hdots & \alpha^{j_{r+1}}& 0 &\hdots & 0\\
      \svdots & \sddots & \svdots & \svdots & \sddots & \svdots \\
      \alpha^{(r-1)j_1} & \hdots & \alpha^{(r-1)j_{r+1}}& 0 &\hdots & 0\\
      0 &\hdots & 0 & 1 &\hdots & 1\vphantom{\alpha^{j_{r+1}'}} \\
      0 &\hdots & 0 & \alpha^{j_1'} & \hdots & \alpha^{j_{r+1}'} \\
      \svdots & \sddots & \svdots & \svdots & \sddots & \svdots \\
      0 &\hdots & 0 & \alpha^{(r-1)j_1'} & \hdots & \alpha^{(r-1)j_{r+1}'}\\
      \alpha^{rj_1} &\hdots & \alpha^{rj_{r+1}} & \alpha^{rj_1'} & \hdots & \alpha^{rj_{r+1}'}\\
      \alpha^{-j_1} &\hdots & \alpha^{-j_{r+1}} & \alpha^{-Nz-j_1'} & \hdots & \alpha^{-Nz-j_{r+1}'}\\
    \end{bsmallmatrix} \ .
  \end{align*}
  To show that the erased positions can be recovered, we need to show that this matrix is invertible. By \cite[Lemma~3]{blaum2016construction} this is true if
  \begin{equation*}
    Nz + \sum_{u=1}^{r+1} j_u' - \sum_{u=1}^{r+1} j_u \neq 0 \mod \order(\beta) \ .
  \end{equation*}
  Note that in \cite[Lemma~3]{blaum2016construction} shows this relation only for $0\leq j_1 <j_2 <\cdots <j_{r+1} \leq n-1$ and $0\leq j_1' <j_2' <\cdots <j_{r+1}' \leq n-1$. However, it is easy to check that the result is independent of the specific values and only depends on the sums $\sum_{u=1}^{r+1} j_u$ and $\sum_{u=1}^{r+1} j_u'$.
  Since all $j_u$ are distinct, we have
  \begin{equation}\label{eq:boundOnSum1}
    \frac{r(r+1)}{2} = \sum_{u=0}^{r} u \leq \sum_{u=1}^{r+1} j_u
  \end{equation}
  and
  \begin{align}
    \sum_{u=1}^{r+1}j_u &\leq \sum_{u=0}^{r}\big((\max_{ j\in \cL}j -r)+u\big) \nonumber \\
                    &= (r+1)(\max_{j\in \cL}j-r) + \sum_{u=0}^{r} u = N-1 + \frac{r(r+1)}{2} \ . \label{eq:boundOnSum2}
  \end{align}
  Combining (\ref{eq:boundOnSum1}) and (\ref{eq:boundOnSum2}) we get
  \begin{align*}
    -(N-1) \leq \sum_{u=1}^{r+1} j_u' - \sum_{u=1}^{r+1} j_u \leq N-1.
  \end{align*}
  Then,
  \begin{align*}
    1=N-(N-1) &\leq Nz + \sum_{u=1}^{r+1} j_u' - \sum_{u=1}^{r+1} j_u \\
    &\leq N(\mu-1) +(N-1)  = N\mu-1 < \order(\beta)
  \end{align*}
  and thus it follows that
  \begin{align*}
   1 \leq N(z-1)+ \sum_{u=1}^{r+1} j_u' - \sum_{u=1}^{r+1} j_u \leq N\mu -1 < \order(\beta) \ .
  \end{align*}
\end{IEEEproof}

By similar arguments we also give a general, sufficient condition on $N$ for the code to be an SD code.
\begin{lemma} \label{lem:SDgeneral}
  Let $\mu,n,r$ and $s=2$ be valid PMDS parameters and $\cL$ be any set of non-negative integers with $|\cL| = n$. Then, the code $\code(\mu ,n,r,2,\cL,N)$ is an SD code for any $N \geq \max_{j \in \cL} j+1$.
\end{lemma}
\begin{IEEEproof}
  The case of $r+2$ erasures in the same local group (horizontal code) is the same as in Lemma~\ref{lem:PMDSgeneral} and \cite[Theorem~5]{blaum2016construction}.  Now consider the case of $r$ column erasures in positions $j_1,\ldots,j_{r} \in \cL$ and two random erasures in local groups $z+1$ and $z'+1$, with $0\leq z< z' \leq \mu-1$ in positions $j,j' \in \cL \setminus \{j_1,\ldots,j_{r}\}$. By the same argument as in \cite[Theorem~5]{blaum2016construction} we need to show that $\beta^{-j} + \beta^{-N(z-z')-j'}$ is invertible. With $1\leq z,z' \leq \mu$ and $0\leq j,j' \leq N-1$ we get
  \begin{align*}
    N(z'-z)+j'-j \geq N +j'-j \geq N - (N-1) > 0
  \end{align*}
  and
  \begin{align*}
    N(z'-z)+j'-j \leq N(\mu -1) + N-1 = N\mu  -1 < \order(\beta) \ .
  \end{align*}
  Combining these we get $1\leq N(z'-z)+j'-j \leq N\mu -1$, so
  \begin{align*}
    N(z'-z)+j'-j \neq 0 \mod \order(\beta)
  \end{align*}
  and it follows that $\beta^{-j} + \beta^{-N(z-z')-j'}$ is invertible.
\end{IEEEproof}

With these generalizations of \cite[Construction~A/B]{blaum2016construction} we are now ready to construct PMDS and SD codes, where each local code is a $d$-MSR code.
\begin{construction}[Locally $d$-MSR PMDS/SD array codes]\label{con:SDcodesLocalRegeneration}
Let $s=2$ and $q,\mu,n,r,d,N \in \mathbb{Z}_{>0}$ be positive integers with
\begin{itemize}
    \item $r \leq n$
    \item $q$ a power of $2$
    \item $q \geq \max\{\mu N, bn\}+1$, where $b=d+1-(n-r)$
    \item $\ell = b^n$
\end{itemize}
 For an element $\beta \in \F_q$ with $\order(\beta) \geq \max\{\mu N,nb\}$ denote $\beta_{i,j} = \beta^{in+j-1}, 0\leq i\leq b-1, 1\leq j\leq n$.

We define the following $[\mu n, \mu(n-r)-2; \ell]_{q^M}$ array code $\code(\mu ,n,r,2,N,d;\ell)_{q}$ as
\begin{align*}
\left\{
  \bC \in \F_{q}^{\ell \times \mu n} \, : \, \bH^{(a)} \bC_{a,:} = \0 \, \forall \, a=0,\dots,\ell-1
\right\}.
\end{align*}
The matrix $\bH^{(a)}$ is defined as
  \begin{align*}
    \bH^{(a)} =
    \begin{bmatrix}
      \bH_0^{(a)} & \mathbf{0} &\hdots & \mathbf{0} \\
      \mathbf{0} & \bH_0^{(a)} &\hdots & \mathbf{0} \\
      \vdots & \vdots &\ddots & \cdots \\
      \mathbf{0}& \mathbf{0} &\hdots & \bH_0^{(a)} \\
      \bH_1^{(a)} & \bH_2^{(a)} &\hdots & \bH_{\mu }^{(a)} \\
    \end{bmatrix} \in \F_{q}^{r \mu +2 \times \mu n} \ ,
    \end{align*}
    where
    \begin{align}\label{eq:parityCheckLocalS2}
      \bH_0^{(a)}  =
      \begin{bmatrix}
        1&1& \hdots & 1 \\
        \beta_{a_1,1} & \beta_{a_2,2} & \hdots & \beta_{a_{n},n}\\
        \vphantom{\int\limits^x}\smash{\vdots} & \vphantom{\int\limits^x}\smash{\vdots} & & \vphantom{1}\smash{\vdots} \\
        \beta_{a_1,1}^{r-1} & \beta_{a_2,2}^{r-1} & \hdots & \beta_{a_{n},n}^{r-1}
      \end{bmatrix} \in \F_{q}^{r \times n} \ ,
  \end{align}
  with $a \in [0,\ell-1]$ and $a = \sum_{i = 1}^{n} a_{i} b^{i-1}$  with $a_i \in [0,b-1]$.
   For $0\leq j\leq \mu -1$ let
  \begin{align*}
    \bH_{j+1}^{(a)}\! = \!
\begin{bmatrix}
      \beta_{a_1,1}^r & \beta_{a_2,2}^r & \hdots & \beta_{a_{n},n}^r \\
      \beta^{-jN} \beta_{a_1,1}^{-1} & \!\beta^{-jN} \beta_{a_2,2}^{-1}  & \hdots & \!\beta^{-jN} \beta_{a_{n},n}^{-1}
    \end{bmatrix}\! \in \!\F_{q}^{2 \times n}  .
  \end{align*}
\end{construction}
It remains to show that the local codes are MSR codes and the conditions under which the code is a PMDS or SD code.
\begin{theorem}\label{thm:PMDSMSR}
Let $\mu,n,r$ and $s=2$ be valid PMDS parameters, $d$ be an integer with $n-r\leq d \leq n-1$, and $q > \max\{\mu N, bn\}$ with
  \begin{equation*}
    N=(r+1)(rn-1-r)+1 \ .
  \end{equation*}
  Then the code $\code(\mu ,n,r,2,N,d;\ell)_{q}$ as in~\cref{con:SDcodesLocalRegeneration} is a locally $d$-MSR $\pmds(\mu,n,r,2,\cW;b^{n})$ code over $\F_q$, as in~\cref{def:locallyMSR}, for $\cW = \{W_1,\ldots,W_\mu\}$ with $W_i = [(i-1)n+1, in]$.
\end{theorem}
\begin{IEEEproof}
  First, note that the $\beta_{i,j}$ in~\cref{con:SDcodesLocalRegeneration} are the (distinct) elements $\beta^0,\beta^1,\ldots,\beta^{bn-1}$. In order for the $\beta^i$ to be distinct, we require $\cO(\beta) \geq bn$, i.e., $q > bn$. Now consider the $j$-th local group. The $a$-th row fulfills the parity check equations given in~\cref{eq:parityCheckLocalS2} and since all elements $\beta_{i,j}$ are distinct, it is immediate that the local group is an $[n,n-r; b^{n}]$ Ye-Barg code as in~\cref{def:yeBarg}.

  For the PMDS property, observe that the $a$-th row, i.e., the row fulfilling the parity-check equations $\bH^{(a)}$, is a code $\code(\mu,n,r,2,\cL^{(a)},N)$ as in~\cref{sec:BlaumSDandPMDS}, where $\cL^{(a)} = \{i-1+(a_i-1)n \ | \ i \in [n]\}$ by definition of the $\beta_{i,j}$. For any $a$ it holds that
  \begin{align*}
    \max_{i \in \cL^{(a)}} i \leq \max_{\substack{i \in \cL^{(a)}\\a\in [0,\ell-1]}} i = rn -1 \ .
  \end{align*}
  By~\cref{lem:PMDSgeneral} a code as in~\cref{sec:BlaumSDandPMDS} is PMDS if $N > (r+1)(\max_{i\in \cL}i -r)$ and the lemma statement follows.
\end{IEEEproof}

\begin{corollary}\label{cor:two_parities_construction}
Let $\mu,n,r$ and $s=2$ be valid PMDS parameters and $d$ be an integer with $n-r\leq d \leq n-1$.
Then, there is a $d$-MSR PMDS code over $\Fq$ with these parameters of field size
\begin{align*}
\mu r(rn-r+n-2)+1 \leq q \leq 2\mu r(rn-r+n-2)
\end{align*}
and subpacketization $\ell = [d+1-(n-r)]^n$.
\end{corollary}

\begin{IEEEproof}
We use Theorem~\ref{thm:PMDSMSR} and derive bounds on the smallest field size $q$ satisfying the bound $q > \max\{\mu N, bn\}$ with $N=(r+1)(rn-1-r)+1 = r(rn-r+n-2)$.

First note that $1 \leq b = d+1-(n-r) \leq r$ for the valid choices of $d$. Furthermore, note that $r \geq 1$ and $n\geq r+1\geq 2$.
Thus, we have
\begin{align*}
  \mu N &= \mu r(rn-r+n-2)\\
&= \mu \Big[rn+\underbrace{r^2(n-1)-2r}_{\geq -1}\Big] \\ %
&\geq \mu(rn-1) \geq rn \geq bn \ .
\end{align*}
Hence, we in fact only require $q > \mu N$. There is a prime power\footnote{Trivially, there is a power of two in this range. Further, by Bertrand's postulate, there is even a prime number within this range.} between $\mu N +1$ and $2\mu N$, which proves the claim.
\end{IEEEproof}

\begin{theorem}\label{thm:SDMSR}
  Let $\mu,n,r$ and $s=2$ be valid PMDS parameters and $q > \max\{rn\mu, bn\}$. Then the code $\code(\mu ,n,r,2,rn,d;\ell)_{q}$ as in~\cref{con:SDcodesLocalRegeneration} is a locally $d$-MSR $\sd(\mu,n,r,s,\cW;b^{n})$ code over $\F_q$, for $\cW = \{W_1,\ldots,W_\mu\}$ with $W_i = [(i-1)n+1, in]$.
\end{theorem}
\begin{IEEEproof}
  The proof follows immediately from the proof of~\cref{thm:PMDSMSR} and~\cref{lem:SDgeneral}.
\end{IEEEproof}

\begin{remark}
It is easy to check that by removing the last row of the parity-check matrix \eqref{eq:PMDS_two_parities_parity_check_matrix} of the PMDS codes in \cite{blaum2016construction}, we obtain a PMDS code with one global parity ($s=1$). By the same operation on all the parity-check matrices for the rows of the $d$-MSR PMDS code in \cref{con:SDcodesLocalRegeneration}, we similarly obtain $d$-MSR PMDS codes with one global parity.
We do not discuss this case in detail since the resulting codes have the same field size as the ones with two global parities.
\end{remark}

\section{Universal PMDS codes with local row-wise MDS MSR codes}\label{sec:universalPMDSconstruction}

In this section, we present a general technique for constructing PMDS codes with MSR local codes, by combining an arbitrary row-wise MDS MSR code (cf.~\cref{def:rowWiseMDS}) with a universal PMDS code family.
The latter notion was first defined in \cite{martinez2019universal}, and we formalize its definition below in~\cref{def:universal_pmds_family}.
Roughly speaking, a universal PMDS code family arises from a PMDS construction in which the local code can be chosen arbitrarily as the $\FqM$-span of an $\Fq$-linear MDS code.
Although the universality requirement seems to be strong, there are several PMDS constructions in the literature that fulfill this property, for instance \cite{rawat2014,martinez2019universal} (cf.~the overview in \cite{martinez2019universal}).
For the construction of \cite{gabrys2018constructions}, we show its universality in~\cref{ssec:Gabrys_et_al_universal_PMDS_construction}.
Hence, some of the PMDS constructions with the smallest field sizes in the literature have this property, which enables the new general construction to achieve rather small field sizes as well.
Note that the PMDS construction with local regeneration in~\cref{sec:constructionS2} is not of the type presented here, since the PMDS family in \cite{blaum2016construction} is not universal (due to strong dependencies between the choice of the local and global parities).

\subsection{A general code construction}

The following definition formalizes the notion of universal PMDS code family, which was introduced in \cite{martinez2019universal}.

\begin{definition}[Universal Partial MDS code family]\label{def:universal_pmds_family}
Let $\mu,n,r,s$ be valid PMDS parameters.
A family of codes is a universal PMDS code family $\pmdsfamily(\mu,n,r,s)$ over $\FqM$ if there is a partition $\mathcal{W} = \{W_1,W_2,\ldots,W_{\mu}\}$ (fixed for the entire family) such that
\begin{itemize}
\item every code $\code \in \pmdsfamily(\mu,n,r,s)$ is a $\pmds(\mu,n,r,s,\cW;1)$ code over $\FqM$ and
\item for any MDS code $\code_{\mathsf{local}}[n,n-r,r+1]$ over $\Fq$, there is exactly one $\code \in \pmdsfamily(\mu,n,r,s)$ such that $\code|_{W_i} = \langle\Clocal\rangle_{\FqM} \simeq \code_{\mathsf{local}}^{\times M}$ %
  for all $i=1,\dots,\mu$ (see~\cref{fig:universal_PMDS_illustration} for an illustration of this property). We denote this unique code by $\family(\code_{\mathsf{local}}) := \code$ (i.e., $\family(\cdot)$ is an injective mapping between the set of MDS codes and the family $\family$).
\end{itemize}
\end{definition}

\begin{figure*}[ht]
\begin{center}
\resizebox{0.7\textwidth}{!}{
\begin{tikzpicture}
\def\n{2.5cm}
\def\cwheight{0.4cm}
\def\yshift{-2cm}

\node[left, align=right] at (-0.1cm, 0.5*\cwheight) {Codeword of $\family(\code_{\mathsf{local}})$ $=$};
\draw (0,0) rectangle node {$\in \langle\Clocal\rangle_{\FqM}\vphantom{^{\FqM}}$} (\n,\cwheight);
\draw (\n,0) rectangle node {$\in \langle\Clocal\rangle_{\FqM}\vphantom{^{\FqM}}$} (2*\n,\cwheight);
\draw (2*\n,0) rectangle node {$\cdots$} (4*\n,\cwheight);
\draw (4*\n,0) rectangle node {$\in \langle\Clocal\rangle_{\FqM}\vphantom{^{\FqM}}$} (5*\n,\cwheight);
\node[right] at (5*\n,0.4*\cwheight) {$\in \FqM^{\mu n}$};
\draw [decorate,decoration={brace,amplitude=10pt},xshift=0pt,yshift=2pt] (0,\cwheight) -- (\n,\cwheight) node [black,midway,yshift=0.6cm] {$W_1$};
\draw [decorate,decoration={brace,amplitude=10pt},xshift=0pt,yshift=2pt] (\n,\cwheight) -- (2*\n,\cwheight) node [black,midway,yshift=0.6cm] {$W_2$};
\draw [decorate,decoration={brace,amplitude=10pt},xshift=0pt,yshift=2pt] (4*\n,\cwheight) -- (5*\n,\cwheight) node [black,midway,yshift=0.6cm] {$W_\mu$};
\draw[|->, >=latex] (2.5*\n,-0.5*\cwheight) -- node[left] {Expand every entry in $\FqM$ as a} node[right] {column vector $\Fq^M$ using a basis of $\FqM$ over $\Fq$} (2.5*\n,\yshift+0.5*\cwheight);

\draw [decorate,decoration={brace,amplitude=10pt},xshift=-2pt,yshift=0pt] (0,\yshift-5*\cwheight) -- (0,\yshift) node [left,black,midway,xshift=-0.4cm] {$M$ rows};
\draw (0*\n,\yshift-0*\cwheight) rectangle node {$\in \Clocal$} (1*\n,\yshift-1*\cwheight);
\draw (0*\n,\yshift-1*\cwheight) rectangle node {$\in \Clocal$} (1*\n,\yshift-2*\cwheight);
\draw (0*\n,\yshift-2*\cwheight) rectangle node[yshift=0.2*\cwheight] {$\vdots$} (1*\n,\yshift-4*\cwheight);
\draw (0*\n,\yshift-4*\cwheight) rectangle node {$\in \Clocal$} (1*\n,\yshift-5*\cwheight);
\draw (1*\n,\yshift-0*\cwheight) rectangle node {$\in \Clocal$} (2*\n,\yshift-1*\cwheight);
\draw (1*\n,\yshift-1*\cwheight) rectangle node {$\in \Clocal$} (2*\n,\yshift-2*\cwheight);
\draw (1*\n,\yshift-2*\cwheight) rectangle node[yshift=0.2*\cwheight] {$\vdots$} (2*\n,\yshift-4*\cwheight);
\draw (1*\n,\yshift-4*\cwheight) rectangle node {$\in \Clocal$} (2*\n,\yshift-5*\cwheight);
\draw (2*\n,\yshift-0*\cwheight) rectangle node[yshift=0.2*\cwheight] {$\ddots$} (4*\n,\yshift-5*\cwheight);
\draw (4*\n,\yshift-0*\cwheight) rectangle node {$\in \Clocal$} (5*\n,\yshift-1*\cwheight);
\draw (4*\n,\yshift-1*\cwheight) rectangle node {$\in \Clocal$} (5*\n,\yshift-2*\cwheight);
\draw (4*\n,\yshift-2*\cwheight) rectangle node[yshift=0.2*\cwheight] {$\vdots$} (5*\n,\yshift-4*\cwheight);
\draw (4*\n,\yshift-4*\cwheight) rectangle node {$\in \Clocal$} (5*\n,\yshift-5*\cwheight);
\node[right] at (5*\n,\yshift-2.5*\cwheight) {$\in \Fq^{M \times (\mu n)}$};

\end{tikzpicture}
}
\end{center}
\caption{Illustration of the codeword structure of the PMDS code $\family(\code_{\mathsf{local}})$ in~\cref{def:universal_pmds_family}.}
\label{fig:universal_PMDS_illustration}
\end{figure*}
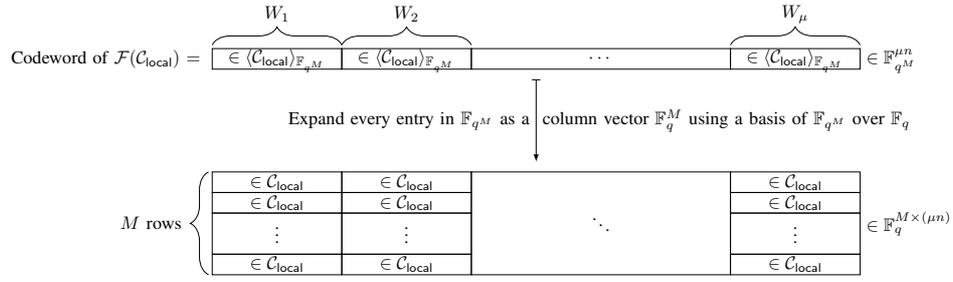

The following code construction combines a universal PMDS code family and a row-wise MSR code.
Note that the code is well-defined since the MSR code is row-wise MDS (i.e., $\family(C_{\msr}^{(a)})$ is well-defined for all rows $a$ of the MSR code).

\begin{construction}\label{constr:general_universal_pmds_row_wise_MSR_construction}
Let $\mu,n,r,s$ be valid PMDS parameters
and $\pmdsfamily(\mu,n,r,s)$ be a universal PMDS code family.
Let $\code_{\msr}[n,n-r;\ell]$ be a row-wise MDS $(h,d)$-MSR code and denote by $\code_{\msr}^{(a)}$ the MDS code in its $a$-th row for $a=0,\dots,\ell$.
We define the code
\begin{align*}
  \pmdsfamily&(\code_{\msr}) \coloneqq \\
  &\left\{
  \bC \in \F_{q^M}^{\ell \times \mu n} \, : \, \bC_{a,:}  \in \family(\code_{\msr}^{(a)}) \, \forall \, a=0,\dots,\ell-1
\right\}.
\end{align*}
\end{construction}

\begin{theorem}\label{thm:universal_construction_general_statement}
  The code $\pmdsfamily(\code_{\msr})$ in~\cref{constr:general_universal_pmds_row_wise_MSR_construction} is a locally $(h,d)$-MSR $\pmds(\mu,n,r,s,\cW;\ell)$ code over $\FqM$, for a partition $\cW = \{W_1,W_2,\ldots,W_{\mu}\}$ of $[\mu n]$ with $|W_i|=n \ \forall \ i\in [\mu]$.
\end{theorem}

\begin{IEEEproof}
By construction, the codewords of $\pmdsfamily(\code_{\msr})$ are matrices whose rows are contained in a PMDS code of the family $\pmdsfamily$. In particular, the PMDS code in the $a$-th row has the MDS code $C_{\msr}^{(a)}$ as its local code. If we puncture all rows in all positions but $W_i$ (for some $i=1,\dots,\mu$), we obtain in the $a$-th row the code $C_{\msr}^{(a)}$. Hence,
\begin{align*}
  \pmdsfamily&(\code_{\msr})|_{W_i} = \\
  &\left\{ \begin{bsmallmatrix}
  \bC_{0,:}^{(0)} \\
  \vdots \\
  \bC_{M-1,:}^{(0)} \\
  \bC_{0,:}^{(1)}\\
  \vdots 	   \\
  \bC_{M-1,:}^{(\ell-1)}
\end{bsmallmatrix} \in \Fq^{M\ell \times n} \, : \, \bC_{j,:}^{(a)} \in C_{\msr}^{(a)} \, \forall \, \substack{a=0,\dots,\ell-1 \\ j=0,\dots,M-1}\right\} \\
&\simeq \underbrace{\code_{\msr} \times \dots \times \code_{\msr}}_{M \text{ times}}  \quad \forall \, i=1,\dots,\mu,
\end{align*}
where the last step follows by re-arranging the rows of a codeword (see~\cref{fig:universal_PMDS_local_regeneration} for an illustration).
Hence, the overall local code is a product of $M$ $d$-MSR codes, and hence is $d$-MSR itself.
The claim follows by the definition of $d$-MSR PMDS codes.
\end{IEEEproof}

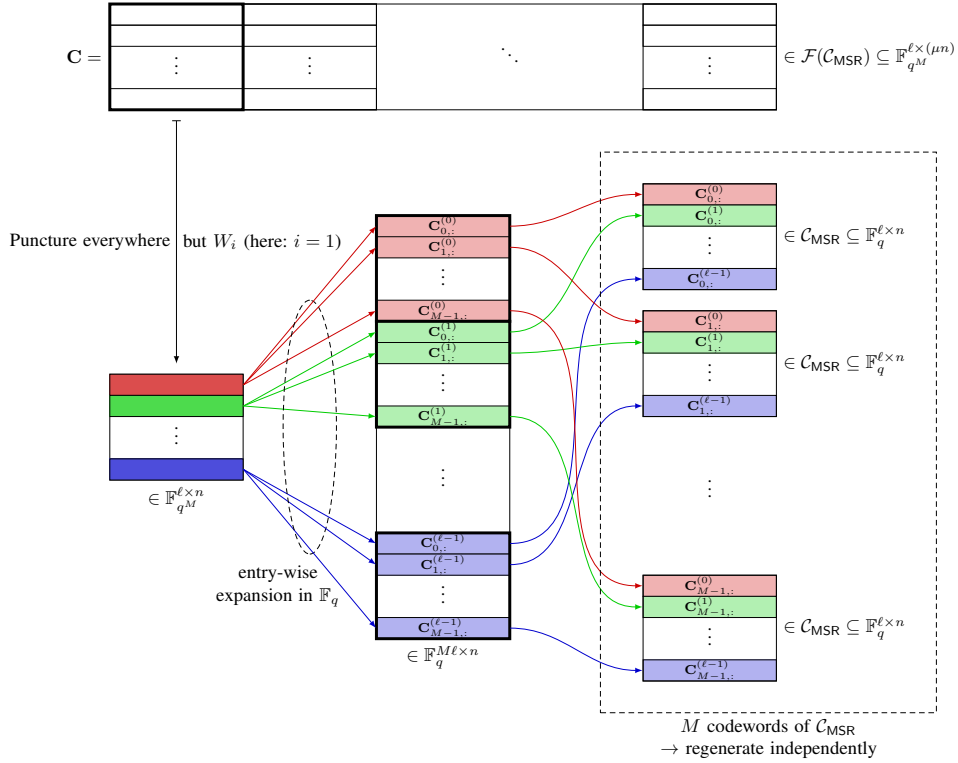
\begin{figure*}[ht]
\begin{center}
\resizebox{0.7\textwidth}{!}{\begin{tikzpicture}
\def\n{2.5cm}
\def\cwheight{0.4cm}
\def\yshift{-7cm}
\def\xshiftexp{2*\n}
\def\xshiftsort{4*\n}
\def\yshiftexp{\yshift+2.5*\cwheight}
\def\yshiftsortone{\yshift+9*\cwheight}
\def\yshiftsorttwo{\yshift+3*\cwheight}
\def\yshiftsortdots{\yshift-3*\cwheight}
\def\yshiftsortthree{\yshift-9.5*\cwheight}
\definecolor{cw1expcolor}{rgb}{0.8,0,0}
\definecolor{cw2expcolor}{rgb}{0,0.8,0}
\definecolor{cw3expcolor}{rgb}{0,0,0.8}
\newcommand{\myfontsize}{\scriptsize}

\def\myeps{2*\cwheight}

\tikzstyle{codeword}=[draw,rectangle, minimum width=\n, minimum height=\cwheight, inner sep=0pt]

\node[codeword, minimum height=5*\cwheight, minimum width=5*\n] (cw) at (2*\n,0) {};
\node[right] at (cw.east) {$\in \family(\code_{\msr}) \subseteq \FqM^{\ell \times (\mu n)}$};
\node[codeword] (cw11) at (0,2*\cwheight) {};
\node[codeword] (cw21) at (0,1*\cwheight) {};
\node at (0,-0.15*\cwheight) {$\vdots$};
\node[codeword] (cw31) at (0,-2*\cwheight) {};
\node[codeword, minimum height=5*\cwheight, ultra thick] (cw1) at (0,0) {};
\node[left] at (cw1.west) {$\bC = $};

\node[codeword] (cw12) at (\n,2*\cwheight) {};
\node[codeword] (cw22) at (\n,1*\cwheight) {};
\node at (\n,-0.15*\cwheight) {$\vdots$};
\node[codeword] (cw32) at (\n,-2*\cwheight) {};
\node[codeword, minimum height=5*\cwheight] at (\n,0) {};

\node[codeword] (cw13) at (4*\n,2*\cwheight) {};
\node[codeword] (cw23) at (4*\n,1*\cwheight) {};
\node at (4*\n,-0.15*\cwheight) {$\vdots$};
\node[codeword] (cw33) at (4*\n,-2*\cwheight) {};
\node[codeword, minimum height=5*\cwheight] at (4*\n,0) {};

\node at (2.5*\n,0.35*\cwheight) {$\ddots$};

\draw[|->, >=latex] (0,-3*\cwheight) -- node[left] {Puncture everywhere} node[right] {but $W_i$ (here: $i=1$)} (0,\yshift+3*\cwheight);

\node[codeword, minimum height=5*\cwheight] (cwpunc) at (0,\yshift) {};
\node[codeword, fill=cw1expcolor!70] (cw1punc) at (0,\yshift+2*\cwheight) {};
\node[codeword, fill=cw2expcolor!70] (cw2punc) at (0,\yshift+1*\cwheight) {};
\node at (0,\yshift-0.15*\cwheight) {$\vdots$};
\node[codeword, fill=cw3expcolor!70] (cw3punc) at (0,\yshift-2*\cwheight) {};

\node[below] at (cwpunc.south) {$\in \FqM^{\ell \times n}$};

\node[codeword, minimum height=20*\cwheight] (cwexp) at (\xshiftexp,\yshiftexp-2.5*\cwheight) {};
\node[below] at (cwexp.south) {$\in \Fq^{M\ell \times n}$};

\node[codeword, fill=cw1expcolor!30] (cw11exp) at (\xshiftexp,\yshiftexp+7*\cwheight) {\myfontsize $\bC^{(0)}_{0,:}$};
\node[codeword, fill=cw1expcolor!30] (cw12exp) at (\xshiftexp,\yshiftexp+6*\cwheight) {\myfontsize $\bC^{(0)}_{1,:}$};
\node at (\xshiftexp,\yshiftexp+5*\cwheight-0.15*\cwheight) {$\vdots$};
\node[codeword, fill=cw1expcolor!30] (cw13exp) at (\xshiftexp,\yshiftexp+3*\cwheight) {\myfontsize $\bC^{(0)}_{M-1,:}$};
\node[codeword, minimum height=5*\cwheight, ultra thick] (cw1exp) at (\xshiftexp,\yshiftexp+5*\cwheight) {};

\node[codeword, fill=cw2expcolor!30] (cw21exp) at (\xshiftexp,\yshiftexp+2*\cwheight) {\myfontsize $\bC^{(1)}_{0,:}$};
\node[codeword, fill=cw2expcolor!30] (cw22exp) at (\xshiftexp,\yshiftexp+1*\cwheight) {\myfontsize $\bC^{(1)}_{1,:}$};
\node at (\xshiftexp,\yshiftexp-0.15*\cwheight) {$\vdots$};
\node[codeword, fill=cw2expcolor!30] (cw23exp) at (\xshiftexp,\yshiftexp-2*\cwheight) {\myfontsize $\bC^{(1)}_{M-1,:}$};
\node[codeword, minimum height=5*\cwheight, ultra thick] (cw2exp) at (\xshiftexp,\yshiftexp) {};

\node at (\xshiftexp,\yshiftexp-4.63*\cwheight) {$\vdots$};

\node[codeword, fill=cw3expcolor!30] (cw31exp) at (\xshiftexp,\yshiftexp-10*\cwheight+2*\cwheight) {\myfontsize $\bC^{(\ell-1)}_{0,:}$};
\node[codeword, fill=cw3expcolor!30] (cw32exp) at (\xshiftexp,\yshiftexp-10*\cwheight+1*\cwheight) {\myfontsize $\bC^{(\ell-1)}_{1,:}$};
\node at (\xshiftexp,\yshiftexp-10*\cwheight-0.15*\cwheight) {$\vdots$};
\node[codeword, fill=cw3expcolor!30] (cw33exp) at (\xshiftexp,\yshiftexp-10*\cwheight-2*\cwheight) {\myfontsize $\bC^{(\ell-1)}_{ M-1,:}$};
\node[codeword, minimum height=5*\cwheight, ultra thick] (cw2exp) at (\xshiftexp,\yshiftexp-10*\cwheight) {};

\draw[densely dashed]  (\n, \yshift-2.4cm) arc(-90:90:0.5cm and 2.4cm);

\draw[->, >=latex, cw1expcolor] (cw1punc.east) -- (cw11exp.west);
\draw[->, >=latex, cw1expcolor] (cw1punc.east) -- (cw12exp.west);
\draw[->, >=latex, cw1expcolor] (cw1punc.east) -- (cw13exp.west);

\draw[->, >=latex, cw2expcolor] (cw2punc.east) -- (cw21exp.west);
\draw[->, >=latex, cw2expcolor] (cw2punc.east) -- (cw22exp.west);
\draw[->, >=latex, cw2expcolor] (cw2punc.east) -- (cw23exp.west);

\draw[->, >=latex, cw3expcolor] (cw3punc.east) -- (cw31exp.west);
\draw[->, >=latex, cw3expcolor] (cw3punc.east) -- (cw32exp.west);
\draw[->, >=latex, cw3expcolor] (cw3punc.east) -- (cw33exp.west);
\draw[densely dashed] (\n, \yshift+2.4cm) arc(90:-90:-0.5cm and 2.4cm);
\node[below, align=center, xshift=-0.4cm] at (\n-0.2cm, \yshift-2.5cm) {entry-wise \\ expansion in $\Fq$};

\node[codeword, fill=cw1expcolor!30] (cw11sort) at (\xshiftsort,\yshiftsortone+2*\cwheight) {\myfontsize $\bC^{(0)}_{0,:}$};
\node[codeword, fill=cw2expcolor!30] (cw12sort) at (\xshiftsort,\yshiftsortone+1*\cwheight) {\myfontsize $\bC^{(1)}_{0,:}$};
\node at (\xshiftsort,\yshiftsortone-0.15*\cwheight) {$\vdots$};
\node[codeword, fill=cw3expcolor!30] (cw13sort) at (\xshiftsort,\yshiftsortone-2*\cwheight) {\myfontsize $\bC^{(\ell-1)}_{0,:}$};
\node[codeword, minimum height=5*\cwheight] (cw1sort) at (\xshiftsort,\yshiftsortone) {};
\node[right] at (cw1sort.east) {$\in \code_{\msr} \subseteq \Fq^{\ell \times n}$};

\draw (cw11exp.east) edge[->, >=latex, cw1expcolor, out=0,in=180] (cw11sort.west);
\draw (cw21exp.east) edge[->, >=latex, cw2expcolor, out=0,in=180] (cw12sort.west);
\draw (cw31exp.east) edge[->, >=latex, cw3expcolor, out=0,in=180] (cw13sort.west);

\node[codeword, fill=cw1expcolor!30] (cw21sort) at (\xshiftsort,\yshiftsorttwo+2*\cwheight) {\myfontsize $\bC^{(0)}_{1,:}$};
\node[codeword, fill=cw2expcolor!30] (cw22sort) at (\xshiftsort,\yshiftsorttwo+1*\cwheight) {\myfontsize $\bC^{(1)}_{1,:}$};
\node at (\xshiftsort,\yshiftsorttwo-0.15*\cwheight) {$\vdots$};
\node[codeword, fill=cw3expcolor!30] (cw23sort) at (\xshiftsort,\yshiftsorttwo-2*\cwheight) {\myfontsize $\bC^{(\ell-1)}_{1,:}$};
\node[codeword, minimum height=5*\cwheight] (cw2sort) at (\xshiftsort,\yshiftsorttwo) {};
\node[right] at (cw2sort.east) {$\in \code_{\msr} \subseteq \Fq^{\ell \times n}$};

\draw (cw12exp.east) edge[->, >=latex, cw1expcolor, out=0,in=180] (cw21sort.west);
\draw (cw22exp.east) edge[->, >=latex, cw2expcolor, out=0,in=180] (cw22sort.west);
\draw (cw32exp.east) edge[->, >=latex, cw3expcolor, out=0,in=180] (cw23sort.west);

\node at (\xshiftsort,\yshiftsortdots+0.35*\cwheight) {$\vdots$};

\node[codeword, fill=cw1expcolor!30] (cw31sort) at (\xshiftsort,\yshiftsortthree+2*\cwheight) {\myfontsize $\bC^{(0)}_{M-1,:}$};
\node[codeword, fill=cw2expcolor!30] (cw32sort) at (\xshiftsort,\yshiftsortthree+1*\cwheight) {\myfontsize $\bC^{(1)}_{M-1,:}$};
\node at (\xshiftsort,\yshiftsortthree-0.15*\cwheight) {$\vdots$};
\node[codeword, fill=cw3expcolor!30] (cw33sort) at (\xshiftsort,\yshiftsortthree-2*\cwheight) {\myfontsize $\bC^{(\ell-1)}_{M-1,:}$};
\node[codeword, minimum height=5*\cwheight] (cw3sort) at (\xshiftsort,\yshiftsortthree) {};
\node[right] at (cw3sort.east) {$\in \code_{\msr} \subseteq \Fq^{\ell \times n}$};

\draw (cw13exp.east) edge[->, >=latex, cw1expcolor,out=0,in=180] (cw31sort.west);
\draw (cw23exp.east) edge[->, >=latex, cw2expcolor,out=0,in=180] (cw32sort.west);
\draw (cw33exp.east) edge[->, >=latex, cw3expcolor,out=0,in=180] (cw33sort.west);

\draw[densely dashed] (3.5*\n-2*\cwheight, \yshiftsortone+4*\cwheight) rectangle (5.7*\n, \yshiftsortthree-4*\cwheight);
\node[below, align=center] at (4.6*\n-\cwheight, \yshiftsortthree-4*\cwheight) {$M$ codewords of $\code_{\msr}$ \\ $\rightarrow$ regenerate independently};

\end{tikzpicture}}
\end{center}
\caption{Illustration of the local regeneration procedure implied by the proof of~\cref{thm:universal_construction_general_statement} (notation as in the proof).}
\label{fig:universal_PMDS_local_regeneration}
\end{figure*}

The remaining difficulty in~\cref{constr:general_universal_pmds_row_wise_MSR_construction} is to find suitable constructions of universal PMDS code families.
In fact, some families in the literature already have this property: the Gabidulin-code-based construction of PMDS codes in \cite{rawat2014} and the PMDS code family constructed from linearized Reed--Solomon (RS) codes in \cite{martinez2019universal} are both universal.
For another construction in the literature, \cite{gabrys2018constructions}, we first show that it can be turned into a universal PMDS code family.
We then summarize the resulting parameters and field sizes for all three specific constructions.

\subsection{Construction~{\ref{constr:general_universal_pmds_row_wise_MSR_construction}} using the Gabidulin-Code-Based PMDS Family} \label{sec:constrGabidulin}

The PMDS code construction in \cite{rawat2014} is based on Gabidulin codes (see \cref{ssec:Gabidulin_codes}), where the fact that the codes have maximal minimum rank distance is used to show that the constructed codes are PMDS.
The construction works as follows.
\begin{itemize}
\item Choose an arbitrary MDS code $\Clocal[n,n-r,r+1]$ over $\Fq$ and a generator matrix $\bG_\mathsf{local}$ thereof.
\item Choose a Gabidulin code $\CGab$ (cf.~ \cref{ssec:Gabidulin_codes}) of parameters $[\mu (n-r), \mu (n-r)-s]$ over $\FqM$. This requires $M \geq \mu (n-r)$.
\item Encode a message in $\FqM^{\mu (n-r)-s}$ with the Gabidulin code $\CGab$, which gives a vector $\bx \in \FqM^{\mu (n-r)}$.
\item Split the vector $\bx$ into $\mu$ sub-blocks $\bx^{(i)}$ of size $(n-r)$, i.e., $\bx = [\bx^{(1)}, \dots, \bx^{(\mu)}]$.
\item Encode each subblock with the generator matrix $\bG_\mathsf{local}$ to obtain the final codeword $\bc$, i.e.,
\begin{align*}
\bc = \left[ \bx^{(1)}\bG_\mathsf{local} , \dots, \bx^{(\mu)}\bG_\mathsf{local} \right] \in \FqM^{\mu n}.
\end{align*}
\end{itemize}
As $\Clocal$ is an arbitrary MDS code, we obtain a universal PMDS code family by fixing a Gabidulin code $\CGab$ and varying the local code.
For fixed PMDS code parameters, the construction requires only $M \geq \mu (n-r)$ (due to the Gabidulin code) and no further restriction on $q$ (other than coinciding with the field size of the local code).
Combining this family with the Ye--Barg MSR codes (which are row-wise MDS),~\cref{thm:universal_construction_general_statement} implies the following statement.

\begin{corollary}\label{cor:universal_construction_Gabidulin}
For all valid PMDS parameters $\mu,n,r,s$, integer $d$ with $n-r \leq d \leq n-1$,
and $\cW = \{W_1,W_2,\ldots,W_{\mu}\}$ a partition of $[\mu n]$ with $|W_i|=n \ \forall \ i\in [\mu]$,
there is a $d$-MSR $\pmds(\mu,n,r,s,\cW;\ell)$ code over $\FqM$ if the field size and subpacketization satisfies
\begin{align*}
M \geq \mu(n-r), \quad q \geq bn, \quad \text{and} \quad  \ell &= b^n,
\end{align*}
where $b = d+1-n+r$. In particular, such a code exists for a field of size $q^M = \left[(d+1-n+r)n\right]^{\mu(n-r)}$.
\end{corollary}

\subsection{Construction~{\ref{constr:general_universal_pmds_row_wise_MSR_construction}} using the Linearized-RS-Codes-Based PMDS Family}

The PMDS code construction in \cite{martinez2019universal} is based on linearized Reed--Solomon codes, which are sum-rank-metric codes that were introduced in \cite{martinez2018skew} and can be seen as a combination of Reed--Solomon and Gabidulin codes.
We do not formally introduce the codes here, but briefly summarize some of their key properties.
Let $\mu <q$, $n' \leq M$, and $k' \leq n' \mu$. Then, a linearized Reed--Solomon code of parameters $[\mu n',k']$ over $\FqM$ is a $k'$-dimensional subspace of $\FqM^{\mu n'}$.
The codes are considered in the sum-rank metric w.r.t.\ the parameter $\mu$, in which codewords are subdivided into $\mu$ blocks of size $n'$ and the distance of two codewords is the sum of the rank distances of the $\mu$ blocks.
The minimal sum-rank distance of a linearized Reed--Solomon code is $\mu n' -k'+1$.
This property is again essential for the codes in \cite{martinez2019universal} to be PMDS.
The construction works as follows.
\begin{itemize}
\item Choose an arbitrary MDS code $\Clocal[n,n-r,r+1]$ over $\Fq$ and a generator matrix $\bG_\mathsf{local}$ thereof.
\item Choose a linearized Reed--Solomon code $\CLRS$ (cf.~\cite{martinez2018skew,martinez2019universal}) of parameters $[\mu (n-r), \mu (n-r)-s]$ over $\FqM$. This requires $M \geq n-r$ and $q > \mu$.
\item Encode a message in $\FqM^{\mu (n-r)}$ with the linearized Reed--Solomon code $\CLRS$, which gives a vector $\bx \in \FqM^{\mu (n-r)}$.
\item Split the vector $\bx$ into $\mu$ sub-blocks $\bx^{(i)}$ of size $(n-r)$, i.e., $\bx = [\bx^{(1)}, \dots, \bx^{(\mu)}]$.
\item Encode each subblock with the generator matrix $\bG_\mathsf{local}$ to obtain the final codeword $\bc$, i.e.,
\begin{align*}
\bc = \left[ \bx^{(1)}\bG_\mathsf{local} , \dots, \bx^{(\mu)}\bG_\mathsf{local} \right] \in \FqM^{\mu n}.
\end{align*}
\end{itemize}
As $\Clocal$ is an arbitrary MDS code, we obtain a universal PMDS code family by fixing a linearized Reed--Solomon code $\CLRS$ and varying the local code.
For fixed PMDS code parameters, the construction requires only $M \geq n-r$ and $q > \mu$ in addition to $q$ being equal to the field size of the local code.\footnote{Note that, compared to the Gabidulin-based PMDS construction above, the restriction on $M$ is much weaker, and we have an additional condition on $q$. This means that the logarithm of the field size (and thus the soft-O complexity of operations in the field) is not linear in $\mu$ anymore, but logarithmic in $\mu$.}
Combining this family with the Ye--Barg MSR codes (which are row-wise MDS),~\cref{thm:universal_construction_general_statement} implies the following statement.

\begin{corollary}\label{cor:universal_construction_sum_rank}
For all valid PMDS parameters $\mu,n,r,s$, integer $d$ with $n-r \leq d \leq n-1$, and $\cW = \{W_1,W_2,\ldots,W_{\mu}\}$ a partition of $[\mu n]$ with $|W_i|=n \ \forall \ i\in [\mu]$,
there is a $d$-MSR $\pmds(\mu,n,r,s,\cW;\ell)$ code over $\FqM$ if the field size and subpacketization satisfies
\begin{align*}
M \geq n-r, \quad q \geq \max\{bn,\mu+1\}, \quad \text{and} \quad  \ell &= b^n,
\end{align*}
where $b = d+1-n+r$. In particular, such a code exists for a field of size $q^M = \max\!\left\{(d+1-n+r)n, \mu+1\right\}^{n-r}$.
\end{corollary}

\subsection{Construction~{\ref{constr:general_universal_pmds_row_wise_MSR_construction}} using the PMDS Family in Gabrys et al.}\label{ssec:Gabrys_et_al_universal_PMDS_construction}

The PMDS code construction in \cite[Section~IV.A]{gabrys2018constructions} uses Reed--Solomon codes as its local codes.
The following theorem generalizes the construction to arbitrary local codes, showing that the code family is in fact universal.
Note that we heavily rely on ideas from \cite[Lemma~2]{gabrys2018constructions}, \cite[Corollary~5]{gabrys2018constructions}, and \cite[Lemma~7]{gabrys2018constructions} in the proof.

\begin{theorem}[Generalization of the PMDS Construction in \cite{gabrys2018constructions}]\label{thm:generalization_PMDS_Gabrys_etal}
Let $n,\mu,r,s$ be valid PMDS parameters and $\FqM$ be a field.
Suppose that there are distinct field elements $\alpha_{1,1},\alpha_{1,2},\ldots,\alpha_{\mu,n} \in \FqM$ such that any subset of $(r+1)s$ elements of the $\alpha_{i,j}$ is linearly independent over $\F_q$.
Define
\begin{align*}
\bH^{(j)} =
\begin{bmatrix}
  \alpha_{j,1} & \alpha_{j,2} &  \hdots & \alpha_{j,n} \\
  \alpha_{j,1}^{q} & \alpha_{j,2}^{q} &  \hdots & \alpha_{j,n}^{q} \\
 \vdots & \vdots & \ddots & \vdots \\
 \alpha_{j,1}^{q^{s-1}} & \alpha_{j,2}^{q^{s-1}} &  \hdots & \alpha_{j,n}^{q^{s-1}} \\
\end{bmatrix} \quad \forall \, 1\leq j\leq \mu.
\end{align*}
Then, the $[\mu n,\mu(n-r)-s]_{\FqM}$ code with parity-check matrix
\begin{align*}
   \bH =
     \begin{bmatrix}
      \bH^{(0)} & \mathbf{0} &\hdots & \mathbf{0} \\
      \mathbf{0} & \bH^{(0)} &\hdots & \mathbf{0} \\
      \vdots & \vdots &\ddots & \cdots \\
      \mathbf{0}& \mathbf{0} &\hdots & \bH^{(0)} \\
      \bH^{(1)} & \bH^{(2)} &\hdots & \bH^{(\mu)}  \\
    \end{bmatrix} \in \FqM^{(\mu r + s) \times \mu n}
  \end{align*}
is a PMDS code, where $\bH^{(0)} \in \Fq^{r \times n}$ (note $\Fq \subseteq \FqM$) is a parity-check matrix of an arbitrary $[n,n-r]_q$ MDS code.

\end{theorem}

\begin{IEEEproof}
Let $\bc = \left[ \bc^{(1)}, \dots, \bc^{(\mu)} \right]$ be a codeword of the code, which is divided into $\mu$ blocks $\bc^{(i)} \in \FqM^{n}$.
By definition, for all $i=1,\dots,\mu$, we have
\begin{equation}
\bH^{(0)} {\bc^{(i)}}^\top = \0 . \label{eq:any_parities_first_parity_eq}
\end{equation}
Furthermore, with $\alphaVec_{i} := [\alpha_{i,1},\alpha_{i,2},\dots,\alpha_{i,n}]$, we have
\begin{align}
\sum_{i=1}^{\mu} \alphaVec_{i}^{q^j} {\bc^{(i)}}^\top = 0, \label{eq:any_parities_second_parity_eq}
\end{align}
for all $j=0,\dots,s-1$.
Let $S := [\bs_1,\dots,\bs_\mu]$ be of the form
\begin{align*}
\bs_i = [s_{i,1},\dots,s_{i,r}] \in [n]^r, \quad s_{i,1} < s_{i,2} < \dots < s_{i,r}.
\end{align*}
Denote by $\bar{\bs}_i$ the vector in $[n]^{n-r}$ that contains, again in increasing order, the entries of $[n]$ that are not contained in $\bs_i$.
The positions $\bs_i$ correspond to the puncturing patterns $E_i$ in the definition of PMDS array codes (cf.~\cref{def:pmds}). We need to show that for each such vector $S$, the array code punctured at these positions in each local group, gives an $[\mu n-\mu r,\mu n-\mu r-s]$ MDS code.

For a vector $\bx$ of length $n$, let $\bx_{\bs_i}$ and $\bx_{\sbar_i}$ be the vectors of length $r$ and $n-r$ containing the entries of $\bx$ indexed by the entries of $\bs_i$ and $\sbar_i$, respectively.
Let $\bH$ be a parity-check matrix of an MDS code of length $n$ and dimension $n-r$. Then, the columns of $\bH$ indexed by $\bs_i$, denoted by $\bH_{\bs_i}$, are invertible and we have for any codeword $\bx$ of the code
\begin{equation*}
\0 = \bH \bx^\top = \bH_{\bs_i} \bx_{\bs_i}^\top + \bH_{\bar{\bs}_i} \bx_{\bar{\bs}_i}^\top \quad
\Rightarrow \quad \bx_{\bs_i}^\top = \bH_{\bs_i}^{-1} \bH_{\bar{\bs}_i} \bx_{\bar{\bs}_i}^\top
\end{equation*}

Hence, it directly follows from \eqref{eq:any_parities_first_parity_eq} that
\begin{align*}
\left(\bc^{(i)}\right)_{\bs_i}^\top = {\bH_{\bs_i}^{(0)}}^{-1} \bH_{\bar{\bs}_i}^{(0)} \left(\bc^{(i)}\right)_{\bar{\bs}_i}^\top,
\end{align*}
and by \eqref{eq:any_parities_second_parity_eq} that (note that $\bH^{(0)}$ has entries in $\F_q$, so $\bH^{(0)} = {\bH^{(0)}}^{q^j}$ for any $j$)
\begin{align*}
0 &= \sum_{i=1}^{\mu} \alphaVec_{i}^{q^j} {\bc^{(i)}}^\top \\
&= \sum_{i=1}^{\mu} (\alphaVec_{i})_{\bs_i}^{q^j} \left(\bc^{(i)}\right)_{\bs_i}^\top + (\alphaVec_{i})_{\sbar_i}^{q^j} \left(\bc^{(i)}\right)_{\sbar_i}^\top \\
  &= \sum_{i=1}^{\mu} \Big[\underbrace{ (\alphaVec_i)_{\bs_i} \left(\bH^{(0)}\right)_{\bs_i}^{-1} \left(\bH^{(0)}\right)_{\bar{\bs}_i} + (\alpha_i)_{\sbar_i}}_{=: \, \gammaVec_{\bs_i}} \Big]^{q^{j}} \left(\bc^{(i)}\right)_{\sbar_i}^\top.
\end{align*}
Thus, the vector
$\bc_{S} = \big[\left(\bc^{(1)}\right)_{\sbar_1}, \left(\bc^{(2)}\right)_{\sbar_2}, \dots, \left(\bc^{(\mu)}\right)_{\sbar_\mu} \big]$, which is the codeword punctured at the positions in $S$,
is contained in a code with parity-check matrix
\begin{align*}
\bH_\gamma :=
\begin{bmatrix}
\gammaVec_S^{q^0} \\
\gammaVec_S^{q^1} \\
\vdots \\
\gammaVec_S^{q^{s-1}} \\
\end{bmatrix},
\end{align*}
where
\begin{align*}
{\gammaVec_S} := \Big[ \gammaVec_{\bs_1}, \gammaVec_{\bs_2}, \dots, \gammaVec_{\bs_\mu} \Big] \in \F_{q^M}^{\mu(n-r)}.
\end{align*}

By definition, we have
\begin{align*}
\gammaVec_{\bs_i} =  (\alphaVec_i)_{\bs_i} \left(\bH^{(0)}\right)_{\bs_i}^{-1} \left(\bH^{(0)}\right)_{\bar{\bs}_i} + (\alpha_i)_{\sbar_i}.
\end{align*}
Since $\left(\bH^{(0)}\right)_{\bs_i}^{-1} \left(\bH^{(0)}\right)_{\bar{\bs}_i}$ is an $r \times (n-r)$ matrix, each entry of $\gammaVec_{\bs_i}$, and thus each entry of $\gammaVec_S$, is a linear combination of at most $r+1$ of the $\alpha_{i,j}$. Furthermore, each such linear combination contains, non-trivially, one element from $\alpha_{i,j}$ (namely the corresponding entry in $(\alpha_i)_{\sbar_i}$) that appears only in this linear combination.
Any set of $s$ entries from $\gammaVec_S$ depends on at most $s(r+1)$ of the $\alpha_{i,j}$, which are linearly independent by the independence assumption.
Hence, the $s$ entries from $\gammaVec_S$ are also linearly independent over $\F_q$.
This means that any $s$ columns of the parity-check matrix $\bH_\gamma$ are linearly independent and $\bH_\gamma$ is a parity-check matrix of an $[n\mu-r\mu,n\mu-r\mu-s]_{q^M}$ MDS code.

It remains to show that the local codes equal the $\FqM$-span of the $[n,n-r]_q$ MDS code with parity-check matrix $\bH^{(0)}$. It is clear by construction that the local codes are subcodes of this code. To see that the local codes are equal to this code, consider the code obtained from the PMDS code after puncturing arbitrary $r$ positions in each local group. This is an $[\mu (n-r), \mu (n-r)-s]_{q^M}$ MDS code. Since valid PMDS parameters fulfill $\mu (n-r)-s \geq n-r$, any $n-r$ columns of a generator matrix of the punctured code are linearly independent. In particular, by further puncturing all positions, except for the remaining $n-r$ positions in one local group, we get an $[n-r,n-r]_{q^M}$ MDS code. This proves that all the local codes have dimension $n-r$, and thus the claim.
Hence, the overall code is a PMDS code.
\end{IEEEproof}

As the MDS code over $\Fq$ can be chosen arbitrarily for fixed $\alpha_{1,1},\alpha_{1,2},\ldots,\alpha_{\mu,n} \in \FqM$,~\cref{thm:generalization_PMDS_Gabrys_etal} immediately implies a universal PMDS code family as in~\cref{def:universal_pmds_family}.
By~\cref{thm:universal_construction_general_statement}, we get the following result.

\begin{corollary}\label{cor:universal_construction_small_fields}
For all valid PMDS parameters $\mu,n,r,s$, integer $d$ with $n-r \leq d \leq n-1$, and $\cW = \{W_1,W_2,\ldots,W_{\mu}\}$ a partition of $[\mu n]$ with $|W_i|=n \ \forall \ i\in [\mu]$,
there is a $d$-MSR PMDS array code as in~\cref{constr:general_universal_pmds_row_wise_MSR_construction} of
field size
\begin{align*}
  n\big[d+1-(n-r)\big]& (n\mu)^{s(r+1)-1} \leq q^M \\
  &\leq 2n\big[d+1-(n-r)\big] (2n\mu)^{s(r+1)-1}
\end{align*}
 and subpacketization
\begin{equation*}
    \ell = \big[d+1-(n-r)\big]^n.
\end{equation*}
\end{corollary}

\begin{IEEEproof}
We combine the universal PMDS code family in~\cref{thm:generalization_PMDS_Gabrys_etal} with Ye--Barg codes (cf.~\cref{def:yeBarg}) using~\cref{constr:general_universal_pmds_row_wise_MSR_construction}.
We choose $q$ and $M$ large enough such that we can ensure that suitable field elements $\alpha_{i,j}$ (of the PMDS code family) and $\beta_{i,j}$ (of the Ye--Barg codes) exist.
A sufficient condition for the existence of the $\beta_{i,j}$ is $q \geq n(d+1-(n-r))$.
Thus, we can choose $q$ to be the smallest prime power greater or equal to $n(d+1-(n-r))$, which is at most $q \leq 2n(d+1-(n-r))$ by Bertrand’s postulate.

For the $\alpha_{i,j}$, it is a bit more involved.
By~\cref{thm:generalization_PMDS_Gabrys_etal}, it suffices to find $n\mu$ elements from $\F_{q^M}$ of which any subset of $s(r+1)$ elements is linearly independent.
We use the same idea as in \cite[Lemma~7]{gabrys2018constructions}. Take the columns of a parity-check matrix of a $\Code[n\mu,n\mu-M,s(r+1)+1]_q$ code and interpret each column in $\F_q^M$ as an element of $\F_{q^M}$.
The independence condition is then fulfilled due to the choice of the minimum distance.

The remaining question is for which $M$ and $q$ a code with parameters $[n\mu,n\mu-M,s(r+1)+1]_q$ exists.
We use the result in \cite[Problem~8.9]{roth2006introduction}, which we can reformulate in our terms as follows. For any $n' = q^a-1$, there exists a code with parameters $[n',n'-M,s(r+1)+1]_q$, where
\begin{align*}
    M \leq 1+ \big[s(r+1)-1\big]a.
\end{align*}
Choose $a$ to be the smallest integer with $n' = q^a-1 \geq n\mu$.
Note that there is such an $a$ with $q^a-1 \leq 2n\mu-1$, i.e., $\log_q(n\mu) \leq a \leq \log_q(2n\mu)$.
Hence, there is an $[n',n'-M,s(r+1)+1]_q$ code with $M \leq 1+ \big[s(r+1)-1\big]\log_q(2n\mu)$.
Shortening the codes gives an $[n\mu,n\mu-M,s(r+1)+1]_q$ code with $M \leq 1+ \big[s(r+1)-1\big]\log_q(2n\mu)$.
\end{IEEEproof}

\section{Discussion and Comparison of PMDS Code Constructions with Local Regeneration}\label{sec:discussion}

\subsection{Field Size Comparison}

Table~\ref{tab:comparison_field_sizes} compares the field sizes of the $d$-MSR PMDS constructions in Section~\ref{sec:universalPMDSconstruction} (universal PMDS code construction), for Ye--Barg MSR codes and three universal PMDS families, and Section~\ref{sec:constructionS2} (two parities).
We also compare our new constructions to the only existing construction of $d$-MSR PMDS codes, which was presented in \cite{rawat2014}. For easier reference, we label the five constructions by the letters $\mathsf{A}$--$\mathsf{E}$.

The known Construction $\mathsf{E}$ (see \cite[Construction~1, case ``$(r+\delta-1)\mid n$'']{rawat2014}) first encodes an information word from $\FqM^{\ell \times (\mu(n-r)-s)}$ with an $[\ell \mu(n-r), \ell (\mu(n-r)-s)]_{q^M}$ Gabidulin code. The resulting codeword is then subdivided into $\mu$ groups, each of length $\ell (n-r)$. These subblocks are then independently encoded using a generator matrix of an $[n,n-r;\ell]_q$ $d$-MSR code. This gives a $d$-MSR PMDS array code with subpacketization $\ell$ and field size $q^M$, where the only requirements on $\ell$ and $q$ are the constraints of the MSR code, and we require $M \geq \ell \mu(n-r)$ in order for the Gabidulin code to exist. An advantage of this construction over ours is that it does not require the MSR code to be row-wise MDS. However, the field size is exponential in the subpacketization (i.e., doubly exponential in $n$ for Ye--Barg codes).

\begin{table*}[ht]
\caption{Comparison of Field Sizes of $d$-MSR PMDS array code constructions (parameters: $d,n,\mu,r,s$ such that $r \leq n$, $s \leq (n-r)\mu$, and $n-r \leq d \leq n-1$).}
\label{tab:comparison_field_sizes}
\begin{center}
\def\arraystretch{1.5}
\begin{tabular}{l|l|c|l}
$\star$ &Construction & Restr. & Smallest field size $Q_\star = q^M$ obtained from the construction \\
\hline \hline
$\mathsf{A}$ & Constr.~\ref{con:SDcodesLocalRegeneration} (Corollary~\ref{cor:two_parities_construction}) & $s=2$ &
$\mu r(rn-r+n-2)+1 \leq \QA \leq 2\mu r(rn-r+n-2)$ \\
$\mathsf{B}$ & Constr.~\ref{constr:general_universal_pmds_row_wise_MSR_construction} + Gabidulin-based PMDS + Ye--Barg (Cor.~\ref{cor:universal_construction_Gabidulin}) & -- & $\QB = \left[(d+1-n+r)n\right]^{\mu(n-r)}$ \\
$\mathsf{C}$ & Constr.~\ref{constr:general_universal_pmds_row_wise_MSR_construction} + Lin.~RS-based PMDS + Ye--Barg (Cor.~\ref{cor:universal_construction_sum_rank}) & -- & $\QC = \max\big\{(d+1-n+r)n,\mu+1\big\}^{n-r}$ \\
$\mathsf{D}$ & Constr.~\ref{constr:general_universal_pmds_row_wise_MSR_construction} + ``small fields'' PMDS + Ye--Barg (Cor.~\ref{cor:universal_construction_small_fields}) & --
& $n\big[d+1-n+r\big] (n\mu)^{s(r+1)-1} \leq \QD$
\\
& & & \hfill $\leq 2n\big[d+1-n+r\big] (2n\mu)^{s(r+1)-1}$ \\
\hline
$\mathsf{E}$ & Known construction: \cite[Construction~1]{rawat2014} + Ye--Barg & -- & $\QE = [(d+1-n+r)n]^{(d+1-n+r)^n \mu (n-r)}$
\end{tabular}
\end{center}
\end{table*}

The following theorem states some relations between the minimal field sizes achievable by the five compared constructions. It can be interpreted as follows:
\begin{itemize}
\item Construction $\mathsf{C}$ has always smaller field size than Constructions $\mathsf{B}$ and $\mathsf{E}$.
\item For two global parities, Construction $\mathsf{A}$ has the smallest field size among all constructions (unless $r$ or $\mu$ is very large).
\item For a large number of global or local parities (and $s>2$), Construction $\mathsf{C}$ has the smallest field size among all constructions.
\item For a small number of global (but $s>2$) and local parities, Construction $\mathsf{D}$ has the smallest field size among all constructions.
\end{itemize}

\begin{theorem}\label{thm:comparison}
For all valid PMDS parameters $\mu,n,r,s$ and integers $d$ with $n-r < d \leq n-1$ (we exclude the trivial case $d=n-r$), denote by $\QA,\QB,\QC,\QD,\QE$ the smallest field sizes obtained from the constructions in Table~\ref{tab:comparison_field_sizes}.
\begin{enumerate}[label=(\roman*)]
\item\label{itm:relation_QC_QB_QE} For all parameters, we have $\QC < \QB < \QE$.
\item\label{itm:relation_s=2} For $s=2$, we have $\QA < \QD$.
If in addition, $r <n-3$, and $\mu \leq n^{n-r-3}$, then $\QA < \QC$.
\item\label{itm:range_QC_best} For $s(r+1)+2r-1\geq 2n$, we have $\QC < \QD$.
\item\label{itm:range_QD_best} For $2s(r+1)+r\leq n$, we have $\QD < \QC$.
\end{enumerate}

\end{theorem}

\begin{IEEEproof}
We use $\mathsf{(a)}$ $a \leq a+1 < 3^a$ and $\mathsf{(b)}$ $a^b \geq ab$ for integers $a,b\geq 1$, which can both be proven easily by induction.

\emph{Ad \ref{itm:relation_QC_QB_QE}:} As $d>n-r$, we have $d+1-n-r>1$, so obviously $\QE > \QB$. If $(d+1-n-r)n \geq \mu+1$, it is clear that $\QB >\QC$ (here we use $\mu \geq 2$). In the case $(d+1-n-r)n \geq \mu+1$, we have
\begin{align*}
\QC=(\mu+1)^{n-r} \overset{\mathsf{(a)}}{<} 4^{\mu (n-r)} \leq \left[(d+1-n+r)n\right]^{\mu(n-r)},
\end{align*}
where $(d+1-n+r)n \geq 4$ holds by assumption.

\emph{Ad \ref{itm:relation_s=2}:} %
We have
\begin{align*}
\QA &\leq 2\mu r(rn-r+n-2) \\
&<2\mu n^2 [2(r+1)-1] \\
&\overset{\mathsf{(b)}}{\leq} 2 n (n\mu)^{2(r+1)-1}\\
&\leq n\big[d+1-n+r\big] (n\mu)^{s(r+1)-1}
\leq \QD,
\end{align*}
where we use $d+1-n+r\geq 2$. Furthermore, if also $r <n-3$ and $\mu \leq n^{n-r-3}$, we have
\begin{align*}
  \QA &\leq 2\mu r(rn-r+n-2)\\
        &< 2 \mu n^3 \leq 2 n^{n-r} \\
&\leq [(d+1-n+r)n]^{n-r} \leq \QC.
\end{align*}

\emph{Ad \ref{itm:range_QC_best}:} Denote $b := d+1-n+r$ and recall that $2 \leq b \leq r<n$. We must show $\QD > (bn)^{n-r}$ and $\QD > (\mu+1)^{n-r}$. We start with the first inequality:
\begin{align*}
  \QD &\geq nb (n\mu)^{s(r+1)-1} 
  \geq nb n^{s(r+1)-1}\\
&\geq nb \big(\underbrace{n^2}_{> nb}\big)^{\frac{s(r+1)-1}{2}}\\
&>  \big(nb\big)^{\frac{s(r+1)+1}{2}}
\geq (nb)^{n-r}.
\end{align*}
The second inequality holds since
\begin{align*}
\QD \geq nb (n\mu)^{s(r+1)-1}
> (\mu+1)^{s(r+1)-1}
\geq (\mu+1)^{n-r}.
\end{align*}

\emph{Ad \ref{itm:range_QD_best}:}
Define $\xi := \max\{nb,\mu\}$ (we use $b := d+1-n+r$ with $2 \leq b \leq r<n$ as above). It suffices to show $\QD < \xi^{n-r}$ under the given conditions.
We have
\begin{align*}
\QD \leq 2nb(\underbrace{2n}_{\leq nb \leq \xi}\mu)^{s(r+1)-1}
< \xi^{2s(r+1)}
\leq \xi^{n-r}
\leq \QC.
\end{align*}

This concludes the proof.
\end{IEEEproof}

Figures~\ref{fig:field_size_comparison_n10_mu5}, \ref{fig:field_size_comparison_n15_mu15}, and \ref{fig:field_size_comparison_n30_mu10} plot the field size bounds of Table~\ref{tab:comparison_field_sizes} over the number of local parities $r$ %
for different sets of PMDS code parameters and $d=n-1$.
The plots also illustrate the change of field size for a varying number of global parities $s$: For Construction~$\mathsf{D}$, we include several curves for different $s$. Note that the field sizes of Constructions~$\mathsf{B}$ and $\mathsf{C}$ are independent of $s$, and Construction~$\mathsf{A}$ exists only for $s=2$. Hence, we only need one curve for each of these three constructions.
The field size of Construction~$\mathsf{E}$ (known construction) is way out of the plot range, which is why it is not contained in the figures.
The plots confirm the statements of Theorem~\ref{thm:comparison} for these example parameters.

\begin{figure}[ht!]
\begin{center}
 \begin{tikzpicture}
\pgfplotsset{compat = 1.3}
\begin{axis}[
	legend style={nodes={scale=0.7, transform shape}},
	width = 1\columnwidth,
	height = 0.7\columnwidth,
	xlabel = {{Number of local parities $r$}},
	ylabel = {{Field Size $q^M$ (logarithmic)}},
	xmin = 1.0,
	xmax = 9,
	ymin = 1.0,
	ymax = 4149515568880992958512407863691161151012446232242436899995657329690652811412908146399707048947103794288197886611300789182395151075411775307886874834113963687061181803401509523685376,
	ytick={2^0, 2^100, 2^200, 2^300, 2^400, 2^500, 2^600},
	yticklabels={$2^{0}$, $2^{100}$, $2^{200}$, $2^{300}$, $2^{400}$, $2^{500}$, $2^{600}$},
	legend pos = north west,
	legend cell align=left,
	ymode=log,
	grid=both]

\addplot [name path=Construction_A_lb, draw=none, forget plot] table[row sep=\\] {
1.000000 86.000000 \\
2.000000 261.000000 \\
3.000000 526.000000 \\
4.000000 881.000000 \\
5.000000 1326.000000 \\
6.000000 1861.000000 \\
7.000000 2486.000000 \\
8.000000 3201.000000 \\
9.000000 4006.000000 \\
};

\addplot [name path=Construction_A_ub, solid, color=constructionAcolor, thick, mark=x, mark size=1.5pt] table[row sep=\\] {
1.000000 170.000000 \\
2.000000 520.000000 \\
3.000000 1050.000000 \\
4.000000 1760.000000 \\
5.000000 2650.000000 \\
6.000000 3720.000000 \\
7.000000 4970.000000 \\
8.000000 6400.000000 \\
9.000000 8010.000000 \\
};
\addlegendentry{{Construction~$\mathsf{A}$ (\textbf{only} $s=2$)}};

\addplot[constructionAcolor!20, forget plot] fill between[of=Construction_A_lb and Construction_A_ub];

\addplot [solid, color=constructionBcolor, thick, mark=*, mark size=1.5pt] table[row sep=\\] {
1.000000 999999999999999929757289024535551219930759168.000000 \\
2.000000 10995116277760000334016270611381010577371608280727552.000000 \\
3.000000 5003154509899970980142786282433158606671303048953856.000000 \\
4.000000 1152921504606846998925411587525957310856881504256.000000 \\
5.000000 2980232238769531477384860548135250905006080.000000 \\
6.000000 365615844006297569556580213138128896.000000 \\
7.000000 4747561509942999819332091904.000000 \\
8.000000 10737418240000000000.000000 \\
9.000000 5904900000.000000 \\
};
\addlegendentry{{Construction~$\mathsf{B}$ (any $s \geq 1$)}};

\addplot [solid, color=constructionCcolor, thick, mark=*, mark size=1.5pt] table[row sep=\\] {
1.000000 1000000000.000000 \\
2.000000 25600000000.000000 \\
3.000000 21870000000.000000 \\
4.000000 4096000000.000000 \\
5.000000 312500000.000000 \\
6.000000 12960000.000000 \\
7.000000 343000.000000 \\
8.000000 6400.000000 \\
9.000000 90.000000 \\
};
\addlegendentry{{Construction~$\mathsf{C}$ (any $s \geq 1$)}};

\addplot [name path=A, solid, draw=none, forget plot] table[row sep=\\] {
1.000000 500.000000 \\
2.000000 50000.000000 \\
3.000000 3750000.000000 \\
4.000000 250000000.000000 \\
5.000000 15625000000.000000 \\
6.000000 937500000000.000000 \\
7.000000 54687500000000.000000 \\
8.000000 3125000000000000.000000 \\
9.000000 175781250000000000.000000 \\
};

\addplot [name path=B, solid, color=constructionDcolor, thick, mark=*, mark size=1.5pt] table[row sep=\\] {
1.000000 2000.000000 \\
2.000000 400000.000000 \\
3.000000 60000000.000000 \\
4.000000 8000000000.000000 \\
5.000000 1000000000000.000000 \\
6.000000 120000000000000.000000 \\
7.000000 14000000000000000.000000 \\
8.000000 1600000000000000000.000000 \\
9.000000 180000000000000000000.000000 \\
};
\addlegendentry{{Construction~$\mathsf{D}$ ($s=1$)}};

\addplot[constructionDcolor!20, forget plot] fill between[of=A and B];

\addplot [name path=A, densely dashed, draw=none, forget plot] table[row sep=\\] {
1.000000 1250000.000000 \\
2.000000 6250000000.000000 \\
3.000000 23437500000000.000000 \\
4.000000 78125000000000000.000000 \\
5.000000 244140624999999995904.000000 \\
6.000000 732421874999999942623232.000000 \\
7.000000 2136230468750000059658010624.000000 \\
8.000000 6103515624999999808555060297728.000000 \\
9.000000 17166137695312499829730376624898048.000000 \\
};

\addplot [name path=B, densely dashed, color=constructionDcolor, thick, mark=*, mark size=1.5pt] table[row sep=\\] {
1.000000 20000000.000000 \\
2.000000 400000000000.000000 \\
3.000000 6000000000000000.000000 \\
4.000000 80000000000000000000.000000 \\
5.000000 999999999999999983222784.000000 \\
6.000000 11999999999999999059939033088.000000 \\
7.000000 140000000000000003909747384254464.000000 \\
8.000000 1599999999999999949813857726687608832.000000 \\
9.000000 17999999999999999821459359399829095579648.000000 \\
};
\addlegendentry{{Construction~$\mathsf{D}$ ($s=2$)}};

\addplot[constructionDcolor!20, forget plot] fill between[of=A and B];

\addplot [name path=A, densely dotted, draw=none, forget plot] table[row sep=\\] {
1.000000 19531250000000000.000000 \\
2.000000 12207031249999998954242048.000000 \\
3.000000 5722045898437499943243458874966016.000000 \\
4.000000 2384185791015624872422878617163131999223808.000000 \\
5.000000 931322574615478468488096536206099080819305974071296.000000 \\
6.000000 349245965480804424089196129715476509752726000892262668566528.000000 \\
7.000000 127329258248209946158414764295162226343119332685454272741710713847808.000000 \\
8.000000 45474735088646411657176293006980327989274251399675529511087777214883457662976.000000 \\
9.000000 15987211554602254091255863182514849771071640250139833814846410037018924956292845928448.000000 \\
};

\addplot [name path=B, densely dotted, color=constructionDcolor, thick, mark=*, mark size=1.5pt] table[row sep=\\] {
1.000000 20000000000000000000.000000 \\
2.000000 399999999999999965732603428864.000000 \\
3.000000 5999999999999999940486453133276365193216.000000 \\
4.000000 79999999999999995719222155803854345574979318317056.000000 \\
5.000000 999999999999999949387135297074018866963645011013410073083904.000000 \\
6.000000 11999999999999999337881695712461963155553465765449079847479399682146304.000000 \\
7.000000 139999999999999992138831935449925184808511285378398708820252613393702318409515008.000000 \\
8.000000 1599999999999999991605865203068878917052713767210914809951067115710532936612925522189484032.000000 \\
9.000000 17999999999999999897718261530532710659157074050397725173125021793518205539544203961683388933500567552.000000 \\
};
\addlegendentry{{Construction~$\mathsf{D}$ ($s=5$)}};

\addplot[constructionDcolor!20, forget plot] fill between[of=A and B];

\addplot [name path=A, densely dashdotted, draw=none, forget plot] table[row sep=\\] {
1.000000 1907348632812499885004360907751424.000000 \\
2.000000 372529029846191395702913588138163837976516516380672.000000 \\
3.000000 54569682106375689500374697124389996157824549777796045886247535640576.000000 \\
4.000000 7105427357601001626620762892490939136177510603478108016289821900608784559393369751552.000000 \\
5.000000 867361737988403529861735792751664312548458620700562065980322993949725939230786732913454912851155091456.000000 \\
6.000000 101643953670516039546855471846051578348063020649230710653850750358308503368293990466138783849149859849812805131748507648.000000 \\
7.000000 11580528575742387265577146326149141917123028578298023294218838984664726169485503796371183103342896046623065560338384790481374029750992896.000000 \\
8.000000 1292469707114105716725898962984866781878073952566925822205370102506903643416657233306553466534324559919010071327293222922613649451853064066442857663168512.000000 \\
9.000000 141994962939782115313243780036356686496170365146979008030366407056633877415738063271302111385412205503610798208299123875344935595805991337521535878002608856192081532551168.000000 \\
};

\addplot [name path=B, densely dashdotted, color=constructionDcolor, thick, mark=*, mark size=1.5pt] table[row sep=\\] {
1.000000 1999999999999999879418332743206357172224.000000 \\
2.000000 399999999999999988675152198241856803399745313464708631625728.000000 \\
3.000000 59999999999999994750351335397161031026540056525198869514124416444638292634238976.000000 \\
4.000000 7999999999999999738689350432933530190796265486452443783343694235770366767373822900869918947643752448.000000 \\
5.000000 999999999999999980003468347394201181668805192897008518188648311830772414627428725464789434929992439754776075181077037056.000000 \\
6.000000 119999999999999998652539748382602337863119043700985084834240147436799886550664197270033827984836522920560257056396413488708122675918512586752.000000 \\
7.000000 13999999999999999779222853399355948657679576029124143390968568097231601403057757636977298632007851838112269151912848983288557052567298652953860800181785063325696.000000 \\
8.000000 1599999999999999968728796375704226551326022227947620350626335730401824086181252359348311027751967738390754377567103205414719952050695043910145556653040705124872716499471129733234688.000000 \\
9.000000 179999999999999987753395690197246235478681545710305822948043244804615710231112366996387842909562546301169493353818131772944426897907688275914677551930379115725986735707903862379086997258067628332679168.000000 \\
};
\addlegendentry{{Construction~$\mathsf{D}$ ($s=10$)}};

\addplot[constructionDcolor!20, forget plot] fill between[of=A and B];

\end{axis}
\end{tikzpicture}
\end{center}
\vspace{-0.5cm}
\caption{Comparison of field sizes of Constructions~$\mathsf{A}$--$\mathsf{D}$ for $n=10$, $\mu=5$, and $d=9$. Construction~$\mathsf{E}$ is not shown as it is out of plot range. Lines are upper bounds, shadows indicate lower bounds. The field sizes of Constructions~$\mathsf{B}$ and $\mathsf{C}$ are independent of $s$.}
\label{fig:field_size_comparison_n10_mu5}
\end{figure}
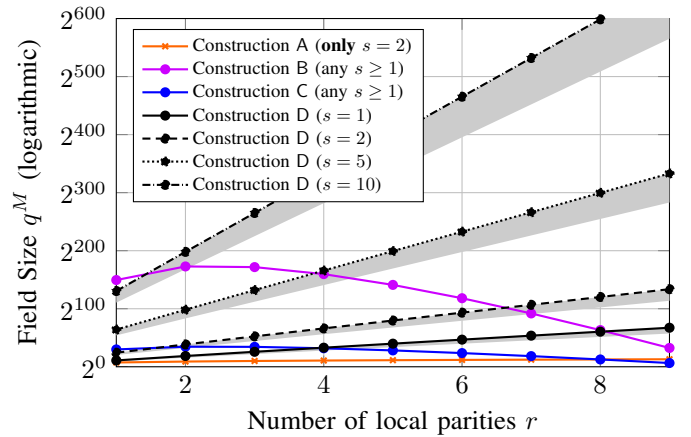

\begin{figure}[ht!]
\begin{center}
 \input{field_size_comparison_n=15_mu=15_d=14_combined_new.tex}
\end{center}
\vspace{-0.5cm}
\caption{Comparison of field sizes of Constructions~$\mathsf{A}$--$\mathsf{D}$ for $n=15$, $\mu=15$, and $d=14$. Construction~$\mathsf{E}$ is not shown as it is out of plot range. Lines are upper bounds, shadows indicate lower bounds. The field sizes of Constructions~$\mathsf{B}$ and $\mathsf{C}$ are independent of $s$.}
\label{fig:field_size_comparison_n15_mu15}
\end{figure}

\begin{figure}[ht!]
\begin{center}
 \input{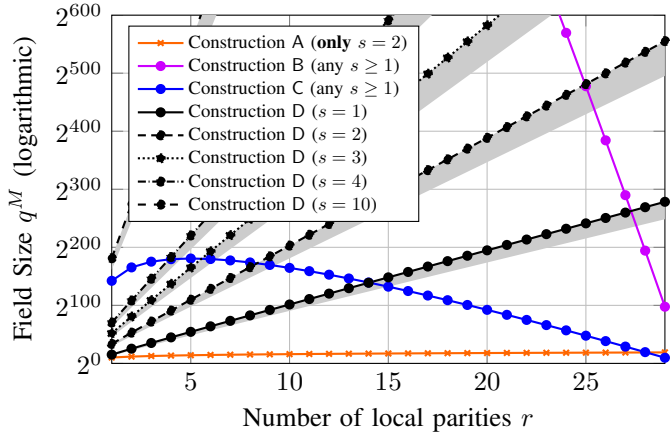}
\end{center}
\vspace{-0.5cm}
\caption{Comparison of field sizes of Constructions~$\mathsf{A}$--$\mathsf{D}$ for $n=30$, $\mu=10$, and $d=29$. Construction~$\mathsf{E}$ is not shown as it is out of plot range. Lines are upper bounds, shadows indicate lower bounds. The field sizes of Constructions~$\mathsf{B}$ and $\mathsf{C}$ are independent of $s$.}
\label{fig:field_size_comparison_n30_mu10}
\end{figure}


\section{PMDS codes with non-trivial global regeneration}\label{sec:global_regeneration}

By definition, a PMDS code punctured in arbitrary $r$ positions per group is an MDS code of distance $s+1$. In the following we construct PMDS codes where each of these MDS codes is an MSR code.
For the sake of simplicity, we focus on the case of highest practical interest: MSR codes with that repair one position ($h=1$) from all remaining positions ($d=\mu(n-r)-1$) of the MDS code.

The construction is based on two main observations: first, the principle used in the MSR codes of \cite{ye2017optimalRepair} can also be applied using Gabidulin codes (recall \cref{def:GabidulinCodes} in Section~\ref{sec:preliminaries}) instead of RS codes; second, performing linearly independent linear combinations of the symbols of a Gabidulin code yields another Gabidulin code with different code locators. Using these observations and carefully choosing the code locators for each row in an array of Gabidulin codewords, we assure the code obtained from puncturing $r$ positions in each local group is MSR.
The construction that we study works as follows.

\begin{construction}[Globally MSR PMDS array codes]\label{con:GloballyRegeneratingPMDS}

  Let $\mu,n,r,s$ be valid PMDS parameters and $\ve{B} \in \FqM^{\ell \times \mu (n-r)}$ be a matrix with entries $\beta_{i,j}, i\in[\ell], j\in [\mu (n-r)]$.
We define the following $[\mu n, \mu(n-r)-s; \ell]_{q^M}$ array code $\code(\mu ,n,r,s,\ve{B};\ell)_{q}$ as
\begin{align*}
\Big\{
  \bC \in \F_{q}^{\ell \times \mu n} \, : \, \bC_{a,:} &= \bu^{(a)} \cdot \bGa_{\ve{B}} \cdot \diag(\bG_{\mathsf{MDS}},\bG_{\mathsf{MDS}},\ldots) \\
  &\forall \, \bu^{(a)} \in \F_{q^m}^{\mu(n-r)-s},  a=0,\dots,\ell-1 \Big\},
\end{align*}
where $\bGa_{\ve{B}}$ is the generator matrix of the code $\Gab(\mu(n-r),\mu(n-r)-s, \ve{B}_{a,:})$ code as in \cref{def:GabidulinCodes} and $\bG_{\mathsf{MDS}}$ is a generator matrix of an $[n,n-r]_{q}$ MDS code.
\end{construction}

It is easy to see that if the rows of the matrix $\ve{B}$ in \cref{con:GloballyRegeneratingPMDS} contain linearly independent elements, then each row of the code is a PMDS code of the code family constructed in \cite{rawat2014}.
In the remainder of this section, we prove that if the matrix $\ve{B}$ is chosen in a suitable way, then the MDS array codes obtained from erasing $r$ positions in each local group are MSR codes of the following type, which can be seen as a Gabidulin-analog of Ye--Barg codes.

\begin{definition}[Skew Ye--Barg $d$-MSR codes]\label{def:skewYeBarg}
  Let $\mu,n,r,s$ be valid PMDS parameters, $\ell \in \ZZ_{>0}$, and $\ve{B} \in \F_{q^m}^{\ell \times \mu (n-r)}$ be a matrix with entries $\ve{B}_{i,j} =\beta_{i,j}$.
  Define $\code(\mu, n, r, s, \ve{B}) \subset \F_q^{\ell \times \mu (n-r)}$ to be an $[\mu (n-r),\mu (n-r) -s;\ell]$ array code over $\F_{q^m}$,
  where each codeword is a matrix with $\ell$ rows and $\mu(n-r)$ columns, where the $a$-th row is a codeword of a code with parity-check matrix
  \begin{equation*}
    \bHa_{\ve{B}} =
      \begin{bmatrix}
        \beta_{a,1} & \beta_{a,2} & \hdots & \beta_{a,{\mu (n-r)}}\\
        \beta_{a,1}^{q^1} & \beta_{a,2}^{q^1} & \hdots & \beta_{a,{\mu (n-r)}}^{q^1}\\
        \vphantom{\int\limits^x}\smash{\vdots} & \vphantom{\int\limits^x}\smash{\vdots} & & \vphantom{1}\smash{\vdots} \\
        \beta_{a,1}^{q^{s-1}} & \beta_{a,2}^{q^{s-1}} & \hdots & \beta_{a,{\mu (n-r)}}^{q^{s-1}}
      \end{bmatrix}\ ,
  \end{equation*}
  for $a \in [0,\ell-1]$. %
  Denote by $\bGa_{\ve{B}}$ a generator matrix corresponding to $\bHa_{\ve{B}}$.
\end{definition}
\begin{remark}
  \cref{def:skewYeBarg} is essentially the same as \cref{def:yeBarg}, except that it relies on Gabidulin codes.
  Note that there is also a difference in presentation: the locators are not given as a set of elements, but instead given explicitly as an input for each row.
  For \cref{def:yeBarg} the corresponding matrix $\ve{B}$ is easily obtained from a set $\cB = \{\beta_{i,j} \}_{i\in [b], j\in [n]}$ of distinct elements of $\FqM$ by assigning $B_{a,j} = \beta_{a_{j},j}$ for $a \in [0,\ell-1]$ and $a = \sum_{i = 1}^{n} a_i b^{i-1}$ with $a_i \in [0,b-1]$, i.e., assigning the code locators of $\bHa$ to the $a$-th row of $\ve{B}$.
\end{remark}

For the node repair algorithm of Ye--Barg codes \cite{ye2017optimalRepair}, it is essential that the rows of a codeword can be partitioned into subsets for which there exist parity checks that differ exactly in position $i$ , i.e., for which all entries are the same except for those at position $i$, which are all distinct (this is due to the close relation of Ye--Barg to Reed--Solomon codes). %
In \cref{lem:YBgroupingGivesMSR} below, we analogously prove that Skew Ye--Barg codes are MSR codes if they have the following property (which is due to their relation to Gabidulin codes).

\begin{definition}[YB-Grouping Property]\label{def:YBgrouping}
  Let $\mu,n,r,s$ be valid PMDS parameters and $\ve{B} \in \F_{q^m}^{\ell \times \mu (n-r)}$. We say that the matrix $\ve{B}$ has the \textbf{YB-grouping property w.r.t.\ $s$} if for each position $i$ the rows of the matrix can be partitioned into $\frac{\ell}{s}$ subsets of $s$ rows for which the elements in the $i$-th position are linearly independent and the elements in all other positions are the same for all $s$ rows.
\end{definition}

\begin{lemma}\label{lem:YBgroupingGivesMSR}
  Let $\mu,n,r,s$ be valid PMDS parameters and $\ve{B} \in \FqM^{\ell \times \mu (n-r)}$ be a matrix such that for any $a \in [\ell]$ the elements of its $a$-th row $\ve{B}^{(a)}$ are linearly independent over $\F_q$. Further, let $\ve{B}$ have the Ye-Barg grouping property w.r.t.\ s as in \cref{def:YBgrouping}. Then the code $\code(\mu, n, r, s, \ve{B})$ as in~\cref{def:skewYeBarg} is an MSR code.
\end{lemma}
\begin{IEEEproof}
  The MDS property follows directly as  each row is a codeword of a Gabidulin code, which are well-known to be MDS. It is easy to check that the recovery algorithm of \cite[Theorem~1]{ye2017optimalRepair} also applies to the code of~\cref{def:skewYeBarg}. For completeness we include a short proof here. Assume node $i$ failed, i.e., we need to recover the set $\{\bC_{a,i} \ \forall \ a \in [\ell] \}$ from the helper nodes with indices $[\mu (n-r)] \setminus \{i\}$.
  Denote by $\mathbb{A}_i = \{\cA_{i,1},\cA_{i,2},\ldots , \cA_{i,\frac{\ell}{s}}\}$ the partition of $[\ell]$ into the subsets $\cA_{i,z}$ of $s$ row indices for which the parity check equations differ exactly in position $i$ and the entries in position $i$ are linearly independent. Note that such a partition exists for every $i\in [\mu (n-r)]$ by definition of the Ye-Barg grouping property.
 The $a$-th row of a codeword $\bC\in \code$ is determined by the $s$ parity checks
  \begin{align*}
    0 = \sum_{j=1}^{\mu (n-r)} \beta_{a,j}^{q^t}\bC_{a,j} = \beta_{a,i}^{q^t}\bC_{a,i} + \sum_{\underset{j\neq i}{j=1}}^{\mu (n-r)} \beta_{a,j}^{q^t}\bC_{a,j}
  \end{align*}
  for $t\in [0,s-1]$. Observe that $\beta_{\cA_{i,z},j} \coloneqq \beta_{a,j} = \beta_{a',j} \ \forall \ a,a' \in \cA_{i,z}, j\neq i$.
  Summing over $a\in \cA_{i,z}$ gives
  \begin{align}
    \sum_{a\in \cA_{i,z}} \beta_{a,i}^{q^t}\bC_{a,i}&= \sum_{a\in\cA_{i,z}} \sum_{\underset{j\neq i}{j=1}}^{\mu n} \Big( \beta_{a,j}^{q^t}\bC_{a,j} \Big) \nonumber \\
    &= \sum_{\underset{j\neq i}{j=1}}^{\mu n} \Big( \beta_{\cA_{i,z},j}^{q^t}\sum_{a\in\cA_{i,z}}\bC_{a,j} \Big) \ . \label{eq:recoverySum}
  \end{align}
  This is a linear system of equations with $s$ unknowns $\bC_{a,i}, a \in \cA_{i,z}$ and $s$ equations, one for each $t \in [0,s-1]$. As the elements $\{\beta_{a,i} \ | \  a \in \cA_{i,z}\}$ are linearly independent by definition (recall that $\ve{B}$ has the Ye-Barg grouping property), the equations are linearly independent. Hence, the unknowns can be uniquely determined if the right hand side of~\cref{eq:recoverySum} is known. Therefore, for repair of node $i$, node $j$ transmits the set of symbols
  \begin{align*}
    \left\{ \sum_{a \in \cA_{i,z}}\bC_{a,j} \ | \ z \in [\ell/s] \right\} \ .
  \end{align*}
  As the cardinality of this set is $\frac{\ell}{s}$, the repair bandwidth is $(\mu (n-r) -1) \frac{\ell}{s}$ and thereby fulfills the bound on the minimal repair bandwidth of \cref{def:regeneratingCode} with equality.
\end{IEEEproof}

The code in \cref{con:GloballyRegeneratingPMDS} can be obtained from a skew Ye--Barg code by multiplying it from the right by the $\mu(n-r) \times \mu n$ matrix $\diag(\bG_{\mathsf{MDS}},\bG_{\mathsf{MDS}},\ldots)$. When puncturing arbitrary $r$ positions in each local group, we do not obtain the original skew Ye--Barg code. However, we do get the original code multiplied by an invertible matrix over $\Fq$ from the right. The rows of such a code are again Gabidulin codes by the following well-known result.
For completeness, we include a short proof of the property.

\begin{lemma}[{\cite[Lemma~3]{berger2003isometries}}]\label{lem:GabidulinCodeLocatorTransformation}
  Let $\bG\in \FqM^{k\times n}$ be a generator matrix of an $[n,k,\dmin]$ Gabidulin code $\Gab(n,k,\bb)$. Then, for any full-rank matrix $\bA \in \Fq^{n\times n}$, the code
  \begin{align*}
    \code' = \left\langle \bG \cdot \bA \right\rangle
  \end{align*}
  is an $[n,k,\dmin]$ Gabidulin code $\Gab(n,k,\bb')$ with
  \begin{align*}
    \bb' =  \bb \cdot \bA^{-1} \ .
  \end{align*}
\end{lemma}
\begin{IEEEproof}
  Let $\bH, \bH' \in \FqM^{n-k \times n}$ be the parity-check matrix of the code $\Gab(n,k,\bb)$, $\Gab(n,k,\bb')$, respectively, as in \cref{def:GabidulinCodes}. By definition we have
  \begin{align*}
    \bG \cdot \bH^T &= \0 \\
    \bG \cdot \bA \cdot \underbrace{\bA^{-1} \bH^T}_{\stackrel{\mathsf{(a)}}{=}\bH'^T} &= \0 \ ,
  \end{align*}
  where $\mathsf{(a)}$ follows from the fact that $\lambda_i \beta_i^{q^l} + \lambda_j \beta_j^{q^l} = (\lambda_i \beta_i + \lambda_j \beta_j)^{q^l} \ \forall \ \lambda_{i},\lambda_j \in \Fq$. As $\bA$ is of full rank over $\Fq$, we have $\rank_q(\bb') = \rank_q(\bb)$ and, in particular, if the elements of $\bb$ are linearly independent, so are the elements of $\bb'$, thereby fulfilling the requirements of~\cref{def:GabidulinCodes} on the code locators.
\end{IEEEproof}

Using the intermediate statements above, the following theorem gives a sufficient condition on the matrix $\ve{B}$ for \cref{con:GloballyRegeneratingPMDS} to give a globally-MSR PMDS code.

\begin{theorem}\label{thm:globalRegeneration}
Let $\mu,n,r,s$ be valid PMDS parameters, $\cW = \{W_1,W_2,\ldots,W_{\mu}\}$ be a partition of $[\mu n]$ with $|W_i|=n \ \forall \ i\in [\mu]$.
Then, the code $\code(\mu ,n,r,s,\ve{B};\ell)_{q^M}$ as in~\cref{con:GloballyRegeneratingPMDS} is a globally-MSR PMDS code if the matrix
\begin{equation*}
\ve{B} \cdot (\diag(\bG_{\mathsf{RS}},\bG_{\mathsf{RS}},...)|_{[\mu n] \setminus \cup_{i=1}^{\mu} E_i})^{-1}
\end{equation*}
has the \emph{YB grouping property} (as in~\cref{def:YBgrouping}) for any $E_i \subset W_i$ with $|E_i|=r$.
\end{theorem}
\begin{IEEEproof}
  Without loss of generality assume that $W_i = [(i-1)n+1,in]$.
  Denote $\cI = {[\mu n] \setminus \cup_{i=1}^{\mu} E_i}$, where $E_i \subset W_i$ with $|E_i|=r$ for all $i\in [\mu]$, and $\bar{E}_i = [n] \setminus E_i$. The restriction of the code $\code$ to the positions indexed by $\cI$ is the code generated by
  \begin{align*}
    \code_{\cI} &= \left\langle \left.\left( \bG^{(a)} \cdot \diag(\bG_{\mathsf{RS}},\bG_{\mathsf{RS}},\ldots) \right)\right|_{\cI}\right\rangle \\
    &= \left\langle \left. \bG^{(a)} \cdot \left(\diag(\bG_{\mathsf{RS}},\bG_{\mathsf{RS}},\ldots) \right)\right|_{\cI}\right\rangle . %
  \end{align*}
  As $\bG_{\mathsf{RS}}$ is the generator matrix of an MDS code, the matrix $\diag(\bG_{\mathsf{RS}}|_{\bar{E}_0},\bG_{\mathsf{RS}}|_{\bar{E}_1},\ldots)$ is a full-rank $\F_q^{\mu(n-r) \times \mu(n-r)}$ matrix. By \cref{lem:GabidulinCodeLocatorTransformation} it follows that code $\code^{(a)}_{\cI}$, consisting of the $a$-th row of every codeword of $\code_{\cI}$, is a $\Gab(\mu(n-r),\mu(n-r)-s,\bb)$ code with
  \begin{align*}
    \bb = \ve{B}_{a,:} \cdot \left.\left(\diag(\bG_{\mathsf{RS}},\bG_{\mathsf{RS}},\ldots\right)\right|_{\cI})^{-1} \ .
  \end{align*}
 As $\ve{B} \cdot \left.\left(\diag(\bG_{\mathsf{RS}},\bG_{\mathsf{RS}},\ldots\right)\right|_{\cI})^{-1}$ has the Ye-Barg grouping property by definition, it follows from \cref{lem:YBgroupingGivesMSR} that the code is MSR.
\end{IEEEproof}

It remains to construct a matrix $\ve{B}$ that fulfills the property of \cref{thm:globalRegeneration}.
We use the following slightly stronger property to simplify the analysis.

\begin{definition}\label{def:YEgroupPropertyWithMatrix}
$\ve{B} \in \FqM^{\ell \times (\mu(n-r))}$ has the \textbf{\scrambledYB property} if $\ve{B} \cdot \diag(\bG_1,\dots,\bG_\mu)$ has the YB grouping property for all invertible matrices $\bG_i \in \Fq^{(n-r) \times (n-r)}$.
\end{definition}

The following theorem gives a construction of a matrix $\ve{B}$ having the \scrambledYB property.

\begin{theorem}\label{thm:construction_matrix_with_scrambled_YB_grouping_property}
Let $M = \mu(n-r+s-1)$ and choose $\mu$ subspaces
\begin{align*}
\cB^{(1)}, \dots, \cB^{(\mu)} \in \mathrm{Gr}(\Fq^{M}, n-r+s-1).
\end{align*}
i.e., $n-r+s-1$-dimensional subspaces of $\Fq^M$, that span the space $\Fq^M$.

For $i=1,\dots,\mu$, consider the sets
\begin{align*}
  \mathcal{S}^{(i)} := \{[\beta_1,&\dots,\beta_{n-r}] \, : \, \langle \beta_1, \dots,\beta_{n-r}\rangle_{\Fq}\\
  &\text{is an $(n-r)$-dimensional subspace of $\cB^{(i)}$}  \}
\end{align*}
and
\begin{align*}
\mathcal{S} := \left\{ [\bb^{(1)} \mid \dots \mid \bb^{(\mu)}] \, : \, \bb^{(i)} \in \mathcal{S}^{(i)}\right\}.
\end{align*}
Then, the cardinality of $\mathcal{S}$ is
\begin{align*}
\ell := |\mathcal{S}| &= \left( \qbin{n-r+s-1}{n-r}{q} \prod_{i=0}^{n-r-1} \big(q^{n-r}-q^{i}\big) \right)^\mu \\
&\leq 4^\mu q^{\mu(n-r)(n-r+s-1)}.
\end{align*}

Let $\ve{B} \in \FqM^{\ell \times (n-r)\mu}$ be a matrix whose rows are exactly the entries of $\mathcal{S}$.
Then, $\ve{B}$ has the \scrambledYB property as in~\cref{def:YEgroupPropertyWithMatrix}.
\end{theorem}

\begin{IEEEproof}
The cardinality of $\mathcal{S}^{(i)}$ is the number of $(n-r)$-dimensional subspaces of an $n-r+s-1$-dimensional vector space over $\Fq$, times the number of bases of such a subspace. The latter equals the number of invertible $(n-r) \times (n-r)$ matrices over $\Fq$. Hence, we have
\begin{align*}
  |\mathcal{S}_i| &= \qbin{n-r+s-1}{n-r}{q} \prod_{i=0}^{n-r-1} \big(q^{n-r}-q^{i}\big) \\
  &\leq 4 q^{(s-1)(n-r)} q^{(n-r)^2} = 4 q^{(n-r)(n-r+s-1)}.
\end{align*}
Overall, we get
\begin{align*}
\ell = |\mathcal{S}| &= \prod_{i=1}^{\mu} |\mathcal{S}^{(i)}| \\
&= \left( \qbin{n-r+s-1}{n-r}{q} \prod_{i=0}^{n-r-1} \big(q^{n-r}-q^{i}\big) \right)^\mu\\
&\leq 4^\mu q^{\mu(n-r)(n-r+s-1)}.
\end{align*}

If we write the elements of $\mathcal{S}$ as a matrix, then this matrix has the YB group property, i.e.,
\begin{itemize}
\item every element of $\mathcal{S}$ is a vector consisting of linearly independent entries (this is obvious since for any choice of the bases $\bb^{(i)}$ their entries are linearly independent by construction of the subspaces $\cB^{(i)}$).
\item for a position $j \in \{1,\dots,n-r\}$ in the $i$-th block ($i=1,\dots,\mu$) and an element $\bb \in \mathcal{S}$, there are the following $s$ elements in $\mathcal{S}$:
Choose  $s-1$ elements $a_2,\dots,a_{s}$ that expand the basis $b^{(i)}_1,\dots, b^{(i)}_{n-r}$ (which spans an $(n-r)$-dimensional subspace) to the $(n-r+s-1)$-dimensional subspace $\cB^{(i)}$.
Then, the $s$ vectors
\begin{align*}
\bb^{(i)} =: \bb^{(i)}_{(1)} &= \begin{bmatrix} b^{(i)}_1 & \dots & b^{(i)}_{j-1} & b^{(i)}_{j} & b^{(i)}_{j+1} & b^{(i)}_{n-r} \end{bmatrix} \\
\bb^{(i)}_{(2)} &= \begin{bmatrix} b^{(i)}_1 &\dots &b^{(i)}_{j-1} & a_2 & b^{(i)}_{j+1} & b^{(i)}_{n-r} \end{bmatrix} \\
&\quad \vdots \\
\bb^{(i)}_{(s)} &= \begin{bmatrix} b^{(i)}_1 &\dots &b^{(i)}_{j-1} & a_s & b^{(i)}_{j+1} & b^{(i)}_{n-r} \end{bmatrix}
\end{align*}
are all in $\mathcal{S}^{(i)}$. Hence, the vectors
\begin{align*}
\Big[\bb^{(1)} \mid \dots \mid \bb^{(i-1)} \mid &\bb^{(i)}_{(1)} \mid \bb^{(i+1)} \mid \dots \mid \bb^{(\mu)} \Big] \\
\Big[\bb^{(1)} \mid \dots \mid \bb^{(i-1)} \mid &\bb^{(i)}_{(2)} \mid \bb^{(i+1)} \mid \dots \mid \bb^{(\mu)} \Big] \\
&\vdots \\
\Big[\bb^{(1)} \mid \dots \mid \bb^{(i-1)} \mid &\bb^{(i)}_{(s)} \mid \bb^{(i+1)} \mid \dots \mid \bb^{(\mu)} \Big]
\end{align*}
are all in $\mathcal{S}$ and differ only in position $j$ in the $i$-th block. Furthermore, the entries $b^{(i)}_{j},a_2,\dots,a_{s}$ in the $j$-th position in the $i$-th block are linearly independent over $\Fq$ by construction.
\end{itemize}

Furthermore, we have $\mathcal{S} = \mathcal{S} \cdot \diag(\bG_1,\dots,\bG_\mu) := \left[ \bb \cdot \diag(\bG_1,\dots,\bG_\mu) \, : \, \bb \in \mathcal{S} \right]$ for all invertible matrices $\bG_i \in \Fq^{(n-r) \times (n-r)}$.
To see this, consider the following: multiplying a subblock $\bb^{(i)}$ with an invertible matrix $\bG_i$ from the right gives another basis of the same subspace---hence $\bb^{(i)}\bG_i \in \mathcal{S}^{(i)}$ and $\bb \cdot \diag(\bG_1,\dots,\bG_\mu) \in \mathcal{S}$ for all $\bb \in \mathcal{S}$. Since the $\bG_i$ are invertible, the mapping $\bb \mapsto \bb \cdot \diag(\bG_1,\dots,\bG_\mu)$ is bijective.

These two observations imply that a matrix with the elements of $\mathcal{S}$ as rows has the \scrambledYB property as in~\cref{def:YEgroupPropertyWithMatrix}.
\end{IEEEproof}

By combining Theorems~\ref{thm:globalRegeneration} and \ref{thm:construction_matrix_with_scrambled_YB_grouping_property}, we get the following existence result for a globally-MSR PMDS code.

\begin{corollary}\label{cor:summary_global_construction}
Let $\mu,n,r,s$ be valid PMDS parameters.
There is a globally-MSR PMDS code with field size
\begin{align*}
(n-1)^{\mu(n-r+s-1)} \leq q^M < [2(n-1)]^{\mu(n-r+s-1)}
\end{align*}
and subpacketization
\begin{align*}
  \ell &= \left( \qbin{n-r+s-1}{n-r}{q} \prod_{i=0}^{n-r-1} \big(q^{n-r}-q^{i}\big) \right)^\mu \\
  &\leq 4^\mu q^{\mu(n-r)(n-r+s-1)}.
\end{align*}
\end{corollary}

\begin{IEEEproof}
The corollary follows directly using the matrix $\ve{B}$ constructed in \cref{thm:construction_matrix_with_scrambled_YB_grouping_property} in Construction~\ref{con:GloballyRegeneratingPMDS} (see \ref{thm:globalRegeneration}).
Choosing $q$ as the smallest prime power $\geq n-1$ ensures that there is an $[n,n-r]_q$ MDS code as required in Construction~\ref{con:GloballyRegeneratingPMDS}. Trivially, there is a power of two with $n-1 \leq q < 2(n-1)$, which proves the claim.
\end{IEEEproof}

\begin{remark}
There are no globally-MSR codes in the literature that we can compare the new construction with.
Therefore, we only compare the field size and subpacketization to a PMDS code without the globally-MSR property, as well as the subpacketization of an MSR code with the same parameters after puncturing $r$ positions in each local group.
I.e., we determine how much we ``pay'' in terms of field size and subpacketization if we go from a purely PMDS or MSR code to a globally-MSR PMDS code.

Construction~\ref{con:GloballyRegeneratingPMDS} is an adaption of the Gabidulin-based PMDS code construction in \cite{rawat2014} (without local or global regeneration), which has field size $q^M < [2(n-1)]^{\mu(n-r)}$.
Compared to such a PMDS code, the exponent in the field size in Corollary~\ref{cor:summary_global_construction} is larger by a factor $1+\tfrac{s-1}{n-r}$.
This difference is significant if the number of global parities is large (recall that $1\leq s \leq \mu(n-r)$).
Hence, we pay more in field size for the globally MSR property if there are many global parities.
It appears possible to adapt other PMDS constructions, such as \cite{martinez2019universal} or \cite{gabrys2018constructions}, to have the globally MSR property as well.
Such a construction may reduce the field size significantly.

Compared to a Ye--Barg MSR code with parameters $[\mu(n-r),\mu(n-r)-s;\ell]$ (which are the code parameters after puncturing $r$ positions in each group of a PMDS code) with subpacketization $[(d+1-n+r)\mu(n-r)]^{\mu(n-r)}$, the subpacketization of the globally-MSR PMDS code in \cref{cor:summary_global_construction} is larger by roughly a factor $(n-r+s-1)$ in the exponent.
Roughly spoken, the exponent of the subpacketization is in $O(\mu n)$ without the PMDS property and in $O(\mu n (n+s))$ for a globally-MSR PMDS code.
\end{remark}

\section{Conclusion}\label{sec:conclusion}

In this paper, we have presented constructions for PMDS array codes with local or global regeneration.
We have presented a construction for local regeneration for two global parities, whose field size is polynomial for a fixed number of local parities.
Furthermore, we have proposed a general construction that combines an arbitrary family of universal PMDS codes with a local row-wise MDS MSR code.
We have explicitly stated the resulting field sizes and subpacketizations for three families of universal PMDS codes, where we first proved the universality property for one existing PMDS construction in the literature.
The presented constructions are based on the PMDS code constructions in \cite{blaum2016construction,rawat2014,martinez2019universal,gabrys2018constructions} and Ye--Barg MSR codes \cite{ye2017optimalRepair}.
All constructions have a significantly smaller field size than the---to the best of our knowledge---only existing construction of PMDS codes with local regeneration: \cite{rawat2014}.
We have also compared the new constructions and identified parameter ranges in which they are best in terms of field size.

Moreover, we have presented a construction of a globally MSR PMDS code by introducing a new MSR code, which can be seen as the skew-analog of Ye--Barg codes (similar to the analogy between Gabidulin and Reed--Solomon codes), and combining it with the Gabidulin-code-based PMDS construction in \cite{rawat2014}.
Compared to the underlying PMDS code, the additional globally-regenerating restriction increases the field size by a factor in the exponent.
Similarly, the PMDS property increases the subpacketization compared to a (global) Ye--Barg MSR code, also by a factor in the exponent.

Several open problems related to the presented results offer interesting opportunities for further research.
Applying the ideas of the presented constructions to the recently proposed PMDS code constructions of~\cite{cai2020construction,gopi2020improved,martinez2020general} could reduce the required field size.
Further, the presented constructions of locally and globally MSR PMDS codes require large levels of subpacketization. The former rely on Ye-Barg regenerating codes, which are known to be suboptimal in terms of subpacketization. However, aside from being optimal in terms of repair bandwidth, they are also row-wise MDS, a property that is essential to the presented constructions. A construction that can afford to relax this requirement could improve the required subpacketization by employing different classes of MSR codes as the local MDS codes. Additionally, the construction of globally MSR PMDS codes is based on Gabidulin codes and thereby inherently suffers from a large required field size. This field size could be lowered by instead employing linearized RS codes to achieve similar gains as shown for locally MSR PMDS codes in \cref{sec:discussion}. Aside from the improvements of the constructions, lower bounds on the required subpacketization and field size would help evaluate the performance of the presented constructions. Finally, for the globally MSR PMDS codes, it remains an open problem to utilize surviving local redundancy nodes, in particular in the extreme case where $r+1$ nodes in a single local group fail while all other nodes survive.

\bibliographystyle{IEEEtran}
\bibliography{main.bib}

\end{document}